\documentclass[10pt]{article}

\usepackage{amsthm,amsmath,amsfonts,enumitem}
\usepackage[dvipsnames]{xcolor}
\usepackage[margin=1in]{geometry}
\usepackage{tikz-cd}
\usepackage{mathtools}
\usepackage[section]{placeins} 

\numberwithin{equation}{section}
\numberwithin{figure}{section}
\numberwithin{table}{section}

\newtheorem{theorem}{Theorem}[section]
\newtheorem{thm}[theorem]{Theorem}
\newtheorem{proposition}[theorem]{Proposition}

\newtheorem{corollary}[theorem]{Corollary}

\theoremstyle{definition}

\newtheorem{example}{Example}[section]

\theoremstyle{remark}
\newtheorem{remark}{Remark}[section]

\newtheorem{implementation goal}{Implementation Goal}

\newcommand{\A}{\mathbb{A}}

\newcommand{\F}{\mathbb{F}}
\newcommand{\FF}{\mathbb{F}} 
\newcommand{\HH}{\mathcal{H}}

\newcommand{\PP}{\mathbb{P}}

\newcommand{\X}{\mathcal{X}}
\newcommand{\Y}{\mathcal{Y}}

\newcommand{\ev}{\operatorname{ev}}

\renewcommand{\ker}{\operatorname{ker}}

\newcommand{\Span}{\operatorname{Span}}
\newcommand{\Tr}{\operatorname{Tr}}
\newcommand{\wt}{\operatorname{wt}}

\title{Minimum Distance and Parameter Ranges of Locally Recoverable Codes with Availability from Fiber Products of Curves}
\author{Mar\'ia Chara, Sam Kottler, Beth Malmskog, \\ Bianca Thompson, and Mckenzie West}

\begin{document}
	\maketitle
	\begin{abstract}    
    We construct families of locally recoverable codes with availability $t\geq 2$ using fiber products of curves, determine the exact minimum distance of many families, and prove a general theorem for minimum distance of such codes.  The paper concludes with an exploration of parameters of codes from these families and the fiber product construction more generally.  We show that fiber product codes can achieve arbitrarily large rate and arbitrarily small relative defect, and compare to known bounds and important constructions from the literature.
	\end{abstract}
  
\section{Introduction}
A code $C$ is broadly said to be locally recoverable if an erased symbol in any position $i$ in a codeword of $C$ can be recovered by consulting a small number of symbols in other (fixed) positions, called a recovery set for position $i$.  Locally recoverable codes have been widely studied in recent years for their potential applications in reliable and efficient cloud storage. For a survey on this topic, see \cite{Balaji2018}.

A natural next property to look for in locally recoverable codes is the ability to recover more than one erasure.  There are two main approaches to this question. First, one could ask that the single recovery set for each position allow for recovery of additional erasures within the set, introducing the parameter $\rho$ to recover $\rho-1$ erasures.  Alternatively, one could ask that each position has multiple (usually disjoint) recovery sets, introducing the parameter $t$ to represent the number of recovery sets that each position has.  Of course, these two approaches can also be blended, producing multiple recovery sets that each can recover multiple erasures.  While this work focuses on the second approach for simplicity, the main construction can be adapted to blend with the first approach.

The Tamo-Barg method \cite{TamoBarg, TamoBargOptimal} of constructing locally recoverable codes is based on building a particular linear space of functions $V$ on an evaluation set $B$. The set $B$ is partitioned into extended recovery sets based on algebraic or geometric relationships between the points in $B$, and functions in $V$ are chosen so that they restrict to polynomials of a single variable of bounded degree on each extended recovery set.  If the value of the function at any point in an extended recovery set is erased, it can be recovered through single variable polynomial interpolation using the values of the function on the other points in the set.  There is a large body of work building on this approach.  In \cite{BTV}, the authors construct locally recoverable codes with availability $t=2$ based on fiber products of curves and propose a group-theoretic perspective on the construction.  In \cite{HMM}, the authors generalize the fiber product construction to $t\geq 2$ and refine the parameters of the resulting codes.  The group-theoretic method of constructing locally recoverable codes with many recovery sets has also been studied, notably in \cite{Bartoli2020}.  The general approach of creating locally recoverable codes from rational maps is pursued in \cite{Munuera2018} and extended to algebraic curves defined by equations with separated variables in \cite{Munuera2020}, but the general fiber product construction still requires more exploration.  

This work is an extension of \cite{HMM}, with a goal of understanding the range of possibilities and limitations of this construction. 
  For completeness, we include the relevant definitions and construction from \cite{HMM}.  In Section \ref{sec:preliminaries}, we include some expository discussion on ways to think of the fiber product of curves and special cases of the construction.  We then introduce the three families of codes which are the main examples of this paper. These three families are all centered on the well-studied Hermitian curve $\HH_q$. The first main example family, introduced in Example \ref{ex:HermitianCurve} comes from the Hermitian curve, introduced as an example of a locally recoverable code with two recovery sets in \cite{BTV}.  The second, Example \ref{ex:HermitianProduct} is a novel code based on the fiber product of two Hermitian curves, and is designed to illustrate the flexibility of this method--one can select curves with appropriate maps and understand the fiber product, and therefore the parameters of the code, using geometry and the construction of \cite{HMM}.  The final example, Example \ref{ex:ASExample}, is a code from a fiber product of Artin-Schreier curves introduced by van der Geer and van der Vlugt. This example was introduced in \cite{HMM} and is included as an example where $t$ can be as large as desired. When the construction is defined over $p^{2h}$, for $p$ a prime and $h$ a natural number, and we choose $t=h$ factor curves, the fiber product is again the Hermitian curve $\HH_{p^t}$.

  In Section \ref{sec:mindist}, we calculate the exact minimum distance for the first family and for a large range of examples in the third family in Theorems \ref{thm:HermitianLRC2Dist} and \ref{thm:ASDistance}.  Incidentally, we also compute the exact minimum distance of a non-fiber product code introduced in \cite{BTV} in Theorem \ref{thm:BTVHqmindist}.
  This is followed by Theorem \ref{thm:MinimumDistance} on minimum distance for codes defined using a fiber product.  We apply Theorem \ref{thm:MinimumDistance} to particular examples in the second and third family in Examples \ref{ex:HermitianDistance2} and \ref{ex:ASExample2}.

  Finally, in Section \ref{SectionBounds}, we explore the parameter space and compare to some relevant bounds and constructions from the literature.  We show in Corollary \ref{cor:RatePC} that fiber product codes are not able to surpass the rate of the product code construction from \cite{TamoBarg}, though some constructions are extremely close. 

%
%

\section*{Acknowledgments}
The authors would like to thank the organizers and participants in the Rethinking Number Theory workshop for creating a lauch pad for part of our research group, and helping us build an inclusive research environment that accommodated all researchers' workloads and backgrounds. We particularly thank early group members Ernest Guico, Darleen Perez-Lavin, and Alexader Diaz-Lopez for their valuable contributions to our understanding. Author Kottler was supported by the Summer Collaborative Research Experience award from Colorado College.  This material is based upon work supported by the National Science Foundation under Grant No. 2137661.

\section{Preliminaries}\label{sec:preliminaries}
  \subsection{Locally Recoverable Codes and Availability}

    Let $n,k$ be natural numbers, with $k\leq n$.
    A \textit{linear code} $C$ of \textit{length} $n$ and \textit{dimension} $k$ over the field $\mathbb{F}_q$ is a $k$-dimensional linear subspace of $(\mathbb{F}_q)^n$.
    The \textit{minimum distance} of $C$ is the minimum number, $d\leq n$, of coordinates in which two distinct elements of $C$ (referred to as \textit{codewords}) must differ.
    The \emph{weight} of a codeword is the number of non-zero coordinates it has, for the codeword $c$, we denote this value by $\wt(c)$.
    As a vector space, the minimum distance of $C$ is equal to the minimum weight of the non-zero codewords.
    It is common to refer to such codes as $[n,k,d]$-codes.
    
    For an $[n,k,d]$-code, the \textit{rate} of the code is  $R=\frac{k}{n}$.  The \textit{relative minimum distance} is given by $\frac{d}{n}$.  When a Singleton-type upper bound $b$ on minimum distance is known, we define the \textit{defect} of the code to be $b-d$ and the \textit{relative defect} of the code to be $\frac{b-d}{n}$.
    
    We say $C$  is a \textit{locally recoverable code (LRC) with locality $r$} if for all $i\in \{1,\dots,n\}$ there exists a set of indices $A_i\subseteq\{1,\dots,n\}\setminus\{i\}$ and a function $\phi_i\colon(\mathbb{F}_q)^r\to \mathbb{F}_q$ such that $\#A_i = r$ and for all codewords $c=(c_1,\dots,c_n)\in C$ we have $c_i =  \phi_i(c|_{A_i})$.
    The set $A_i$ is called the \textit{recovery set} for the $i$-th position.
    It may be desirable to have multiple disjoint recovery sets for each position to protect against multiple erasures or allow for simultaneous queries of heavily-accessed information.
    A locally recoverable code $C$ has \textit{availability} $t$ with locality $(r_1, \dots, r_t)$ if for each $i\in \{1,\dots,n\}$ there exists sets of indices $A_{i,1},\dots,A_{i,t}\subseteq\{1,\dots,n\}\backslash\{i\}$ such that
    \begin{enumerate}
        \item $A_{i,j}\cap A_{i,h} = \emptyset$ for $j\neq h$
        \item $\#A_{i,j} = r_j$
        \item For each $j\in \{1,\dots,t\}$ there exists a function $\phi_{i,j}\colon\F_q^{r_j}\to \F_q$ 
              such that for all codewords $c=(c_1,\dots,c_n)\in C$ we have $c_i = \phi_{i,j}(c|_{A_{i,j}})$.
    \end{enumerate}
    We refer to an LRC with availability $t$ as an LRC($t$). The localities of an LRC($t$) form a vector $(r_1,r_2,\dots, r_t)$.  When $r_i=r_j=r$ for all $i,j\in\{1,2,\dots, t\}$, we say that the code has \textit{uniform locality} $r$.

  \subsection{Evaluation Codes on Curves}

    Let $\X$ be an algebraic variety defined over a finite field $\mathbb{F}_q$. 
    Let $B$ be a subset of $\X(\F_q)$ of cardinality $n\in\mathbb{N}$, with points arbitrarily ordered as $B=\{P_1, P_2,\dots, P_n\}$. 
    Let $V$ be a linear subspace of the function field $\F_q(\X)$ such that no function in $V$ has poles at any point in $B$. 
    For any $f\in V$, define the evaluation map 
    \[ev_B\colon V\rightarrow \F_q^n, \hspace{.25in} f\mapsto (f(P_1), f(P_2),\dots ,f(P_n)).\] 
    Then we define the evaluation code $C(V,B)$ as
    \[ C(V,B)\coloneqq\{ev_B(f):f\in V\}.\]

    Reed-Solomon codes are evaluation codes where $V$ is the space of polynomials of bounded degree and $B$ are the values in a finite field, viewed as affine points on a projective line.
    Evaluation codes on the Hermitian curve have also been very well-studied.
    For any prime power $q$, the Hermitian curve $\HH_q$ is defined over any extension of $\F_q$ by the affine equation
      \[ x^q + x= y^{q+1}. \]
    The curve $\HH_q$ has genus $\frac{1}{2}q(q-1)$ and has $q^3+1$ points over $\F_{q^2}$ including a single point at infinity.

  \subsection{Fiber Products of Curves}\label{sub:FiberProduct}

    Let $\Y_1$, $\Y_2$, and $\Y$ be projective curves over $\F_q$, with maps $h_i\colon \Y_i\rightarrow \Y$ that are separable, rational $\F_q$-morphisms for $i=1,2$.
    The fiber product $\Y_1\times_{\Y}\Y_2$ is a curve that is (abstractly) defined using the corresponding fiber product of schemes.
    More concretely, the $\F_q$-rational points of the fiber product $\Y_1\times_{\Y}\Y_2$ are given by
      \[(\Y_1\times_{\Y}\Y_2)(\F_q)=\{(P_1,P_2)\in \Y_1(\F_q)\times \Y_2(\F_q) : h_1(P_1)=h_2(P_2)\}.\]

    The fiber product construction can be iterated and is seen to be (up to isomorphism) associative and commutative.
    Thus for any $t\in\mathbb{N}$, we may without confusion construct the $t$-fold fiber product of curves as follows.
    Let $\Y,\Y_1,\dots,\Y_t$ be projective curves over $\F_q$ with separable $\F_q$-rational maps $h_j:\Y_j\to \Y$.
    The $\F_q$-points of the fiber product $\X=\Y_1\times_{\Y}\dots\times_{\Y} \Y_t$ of $\Y_1,\dots,\Y_t$ over $\Y$ are then given by
      \[\X(\F_q)= \{(P_1,\dots,P_t): P_i\in\Y_i(\F_q) \text{ and } h_i(P_i) = h_j(P_j) \text{ for all } i,j \in\{1,\dots,t\}\}.\] A simple visualization of a  2-fold fiber product construction is given in Figure \ref{fig:FP visual}.
      
          \begin{figure}[h!]
        \centering
        \includegraphics[scale = 0.3]{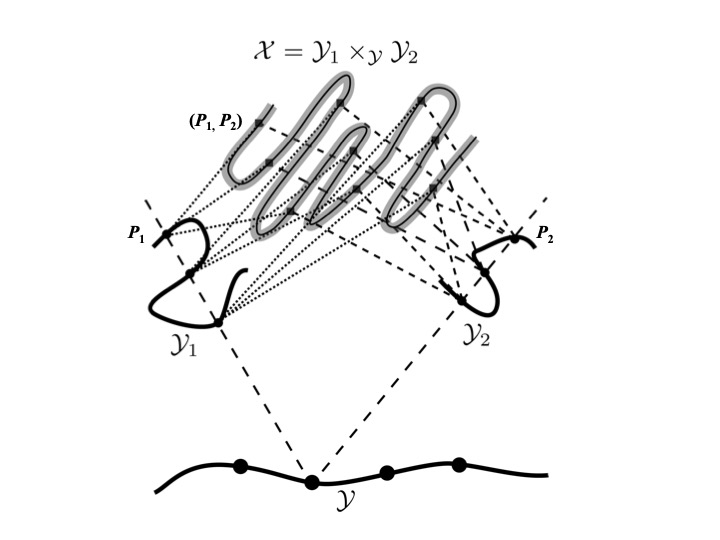}
        \caption{A visualization of the points on a fiber product $\X$ lying above a single point on the base curve $\Y$.}
        \label{fig:FP visual}
    \end{figure}
    This construction induces $t$ natural projection maps
      \[g_i\colon\X\rightarrow \Y_i\] 
    from the fiber product onto each factor curve.
    Let
      \[\tilde{\Y_i}=\Y_1\times_{\Y}\dots \times_{\Y}\Y_{i-1}\times_{\Y}\Y_{i+1}\times_{\Y} \dots\times_{\Y}\Y_t\]
    be the fiber product of all curves $\Y_j$ except $\Y_i$.
    Then we see that $\X$ is isomorphic to $\Y_i\times_{\Y}\tilde{\Y}_i$, and we identify $\tilde{\Y}_i$ with the isomorphic factor in the original fiber product construction of $\X$.
    This gives complementary projection maps
      \[\tilde{g}_i\colon\X\rightarrow \tilde{\Y}_i\quad\text{and}\quad\tilde{h}_i\colon\tilde{\Y}_i\rightarrow \Y.\]  
    We also define the map $g\colon \X\rightarrow \Y$ by $g=h_i\circ g_i$ for any $i$.

    \begin{remark}
        Simply speaking, the map $\tilde{g}_i$ ``forgets" the information coming from the curve $\Y_i$ while retaining the data of the fiber product that come from the other curves.
    \end{remark}

    The function field $\F_q(\X)$ is isomorphic to the compositum of the function fields $\F_q(\Y_i)$, where the function field $\F_q(\Y)$ is embedded into each $\F_q(\Y_i)$ as induced by the map $h_i$.
    For ease of exposition, we identify each function field with its image inside $\F_q(\X)$, so $\F_q(\Y)\subseteq \F_q(\Y_i)\subseteq \F_q(\X)$ for each $i$.
    Further, we assume that $\F_q$ is the full field of constants within each of these fields, and that
    \begin{equation}\label{eqn:trivial_intersection}
      \F_q(\mathcal{Y})=\bigcap_{i=1}^t \F_q(\mathcal{Y}_i).
      \end{equation}\label{rmk:linearly_disjoint}
\begin{remark}If \ref{eqn:trivial_intersection} holds and either each extension $\F_q(\Y_i)/\F_q(\Y)$ is Galois, or the degrees $d_{h_i}$ are pair-wise relatively prime, then these extensions are linearly disjoint and the degree of $\F_q(\X)/\F_q(Y)$ is the product of the degrees $d_{h_i}$, i.e. $d_{g}=\prod_{i=1}^td_{h_i}$. In this case, we have $d_{g_j}=(\prod_{i=1}^td_{h_i})/d_{h_j}$.  \end{remark}




  \subsection{Locally Recoverable Code with Availability $t$ Construction from Fiber Product}\label{sec:CodeConstruction}

    The following general construction comes from \cite{HMM}, though for completeness we include it simplified notation here. 
    Let $y_0\in\F_q(\X)$ so that $\F_q(\Y)=\F_q(y_0)$. For each $i$, $1\leq i \leq t$, we choose $y_i \in \F_q(\X)$ so that $\F_q(\Y_i)=\F_q(\mathcal{Y})(y_i)$, where $y_i$ is the root of an irreducible separable polynomial $b_i(X) \in \F_q(\Y)[X]$. Let $d_{y_i}$ be the degree of the function $y_i\colon\mathcal{X}\rightarrow \mathbb{P}^1_{y_i}$.

    We now have that
    \[\F_q(\mathcal{X})=\F_q(y_0)(y_1)\cdots(y_t)=\F_q(y_0,y_1,y_2,\dots,y_t).\]
    The degree of $\tilde{g}_i$ must be equal to the degree of $h_i$, denoted $d_{h_i}$.

    Now, choose $S\subset \mathcal{Y}(\F_q)$ such that
    \begin{itemize}
    \item $|g^{-1}(P)\cap\mathcal{X}(\F_q)|=d_g$ for all $P\in S$ (i.e. all places in $S$ split completely in the extension $\F_q(\mathcal{X})/\F_q(\mathcal{Y})$) and
    \item for each $i$, $1\leq i\leq t$, the function $y_i$ has no poles at any point above $S$ in the extension $\F_q(\mathcal{X})/\F_q(\mathcal{Y})$.
    \end{itemize}

    Choose an effective divisor $D$ of degree $l$ on $\mathcal{Y}(\F_q)$ with $S \cap \mathrm{supp} (D) = \emptyset$, so functions in the Riemann-Roch  space $\mathcal{L}(D)$ have no poles in $S$.  Let $\{f_1, f_2, \dots, f_m\}$ be a basis of the Riemann-Roch space $\mathcal{L}(D)$.
    We require that $l< |S|$ so that for all $f\in\mathcal{L}(D)$, there exists some $P\in S$ with $f(P)\neq 0$. Let $V$ be the $\F_q$-vector space with basis
      \begin{equation}\label{eqn:Vbasis}
        \{f_j y_1^{{e_1}}\cdots {y_t}^{e_t}: 1\leq j\leq m, 0\leq e_i\leq d_{h_i}-2\textrm{ for all } i\}. 
      \end{equation}
    Then set 
    \begin{equation}\label{eqn:Bdef}
    B=g^{-1}(S)\subset \X(\F_q),
    \end{equation}
    where an arbitrary ordering of elements is fixed on $B$. Note that $n=|B|=d_g|S|$.  

    The code $C(V,B)$ is locally recoverable with availability $t$.
    Recall that we have fixed an ordering of the points in $B$ for the evaluation map $ev_B$.
    For any $i,j\in \mathbb{N}$ with $1\leq i\leq n$ and $1\leq j\leq t$, set
      \[B_{i,j}=\tilde{g}_j^{-1}(\tilde{g}_j(P_i))\setminus \{ P_i \}.\] 
    Let
    \[A_{i,j}=\{a\colon P_a \in B_{i,j}\}.\] 
    Consider a codeword $ev_B(f)$ for some function $f\in V$.
    Given an erasure in position $i$ of the codeword (associated with point $P_i$), each $A_{i,j}$
    acts as a recovery set, because on the set $B_{i,j}$ the function $f$ is constant except in $y_j$, so on $B_{i,j}$ it acts as $\tilde{f}(y_j)$, a polynomial of degree less than or equal to $d_{h_j}-2$.
    The evaluation of $f$ on the $d_{h_j}-1$ points of $B_{i,j}$ therefore give rise to $d_{h_j}-1$ distinct pairs $\left(y_j(P_i), \tilde{f}\left(y_j(P_i)\right)\right)$.
    Since any polynomial of this degree is determined by its values on $d_{h_j}-1$ points, these pairs are sufficient to determine the value of $\tilde{f}(P_i)=f(P_i)$.

    This construction gives rise to the following theorem.
    \begin{theorem} \label{thm:construction}
        Given a fiber product $\X$ of curves defined over $\FF_q$ as described in Section \ref{sub:FiberProduct}, with $V$ a vector space of functions on $\X$ with basis as in \eqref{eqn:Vbasis} and $B$ a subset of $\X(\FF_q)$ as in \eqref{eqn:Bdef}, the code $C(V,B)$ is a locally recoverable code with availability $t$ and 
        \begin{itemize}
          \item length $n=|B|$,
          \item dimension $m(d_{h_1}-1)(d_{h_2}-1) \cdots (d_{h_t}-1)$,
          \item minimum distance $d \geq n-ld_g-\sum_{i=1}^t\left( d_{h_i}-2\right)d_{y_i}$, and
          \item locality $(d_{h_1}-1, d_{h_2}-1, \dots,d_{h_t}-1)$.
        \end{itemize}
    \end{theorem}
    
One may easily calculate the rate $R$ of the constructed code.  In the case that the extensions are linearly disjoint, we have an especially simple form.
\begin{corollary}\label{cor:rate}
If the extensions $\FF_q(\Y_i)/\FF_q(\Y)$ are linearly disjoint, then the rate of $C(V,B)$ is \begin{equation}\label{eqn:LinDisRate}
R=\frac{m}{\left\vert S\right \vert}\prod_{i=1}^t\frac{d_{h_i}-1}{d_{h_i}}.
\end{equation}
\end{corollary}

This is a simple application of the definition of rate, the fact that $\left\vert B\right\vert=d_g\left\vert S\right\vert$, and the fact that when the extensions are linearly disjoint, we have $d_{g}=\prod_{i=1}^td_{h_i}$.

\section{Simplified Framework and Featured Constructions}\label{simple example} 
  To gain some intuition, let us consider the simplest version of this fiber product construction: say $\mathcal{Y}=\mathbb{P}^1_{y_0}$ with $\infty_{\Y}$ the unique point at infinity on this curve,
  and $h_i\colon \mathcal{Y}_i\rightarrow \mathcal{Y}$ given by projection onto $y_0$.
  In this case, the fiber product 
    $\mathcal{X}=\mathcal{Y}_1\times_{\mathcal{Y}}\dots \times_{\mathcal{Y}}\mathcal{Y}_t$
  can be embedded into $\mathbb{P}^{t+1}$, with affine coordinates $(y_0, y_1, \dots, y_t)$.
  Note that this fiber product, $\mathcal{X}$, is isomorphic to the intersection of $t$ hypersurfaces in $(t+1)$-dimensional space.
  Further, if we take $D=l\infty_{\Y}$ to be the divisor defining the Riemann-Roch space $\mathcal{L}(D)$, then this fiber product construction results in a punctured subcode of the Reed-Muller code, with functions simply polynomials in $\mathbb{F}_q[y_0,y_1,\dots, y_t]$ and evaluation points a subset of points on the intersection of the $t$-hypersurfaces created by considering the defining equations for the $t$ curves $\Y_i$ in $\PP^{t+1}$.
  Explicitly, the functions leading to codewords are
    \[V=\textrm{Span}\{y_0^jy_1^{e_1}y_2^{e_2}\cdots y_t^{e_t}:0\leq j\leq l, 0\leq e_i\leq d_{h_i}-2\}.\]
  General fiber product codes should be viewed as generalizations of these simple codes.

  Let $P=(\alpha,\beta_1,\dots, \beta_t)$ be an evaluation point of such a simple fiber product code, where $\alpha,\beta_1,\dots, \beta_t\in\mathbb{F}_q$.
  The $i$-th recovery set for $P$ is the set of all evaluation points 
    $Q=(\alpha, \beta_1,\dots, \beta_{i-1}, \gamma, \beta_{i+1}, \dots, \beta_t)$, 
  where $\gamma\in\mathbb{F}_q$.  That is, the $i$-th recovery set is simply the set of all evaluation points which share all coordinate values but that of $y_i$ with $P$.

  We now introduce three important examples of fiber product codes within this simplified framework.

  \begin{example}{\textbf{LRC(2)s on the Hermitian Curve Viewed as a Fiber Product}}\label{ex:HermitianCurve}
    As a first concrete example, we consider the Hermitian curve $\HH_q$ as a fiber product and intersection.
    Let $\Y=\PP^1$, $\Y_1\colon u=y^{q+1}$, and $\Y_2\colon u=x^q+x$, and let $h_i\colon \Y_i\to\PP^1$ be projection onto $u$ for $i=1,2$.
    Then the fiber product $\X=\Y_1\times_{\Y}\Y_2$ is isomorphic to the curve $H_q\colon x^q+x=y^{q+1}$.
    Indeed, the affine points of $\X(\F_q)$ are given by
    \[\{((y,u),(x,u)):x,y,u\in \F_q, y^{q+1}=u=x^q+x\} \subseteq \PP^2\times\PP^2.\]
    Hence this is isomorphic by the natural map to the intersection of the two hypersurfaces in $\PP^3$ with affine equations $u=x^q+x$ and $u=y^{q+1}$, and also to the curve $\HH_q$ defined in $\PP^2$ by affine equation $y^{q+1}=x^q+x$.
    The utility of the fiber product viewpoint on this curve is to highlight two natural maps which give rise to recovery sets.
    Codes using the fiber product construction of $\HH_q$ are developed in \cite{BTV}, where a lower bound is given on the minimum distance.
    Let $C_{\mathcal{H}_q}$ be the LRC($2$) presented in Proposition 5.1 of \cite{BTV}.
    For this code, we take the curve $\HH_q$ with evaluation set $B_{\mathcal{H}_q}=\{P\in\HH_q(\F_{q^2})\colon y(P)\neq 0 \}.$
    We can check that $|B_{\mathcal{H}_q}|=q^3-q$.
    Then we let $V_{\mathcal{H}_q}$ be the space of functions with basis $\{x^iy^j\colon 0\leq i \leq q-2, 0\leq j\leq q-1\}.$
    The code $C_{\mathcal{H}_q}=C(V_{\mathcal{H}_q}, B_{\mathcal{H}_q})$ is an LRC($2$) where the two recovery sets for a point $P\in B_{\mathcal{H}_q}$ are given by the points $Q \in B_{\mathcal{H}_q}$, $Q\neq P$ sharing the same $x$-coordinate $P$ and those sharing the same $y$-coordinate value as $P$.
    These recovery sets are of size $q-1$ and $q$, respectively.

    In \cite{BTV}, the authors prove the following.
    \begin{thm} (\cite{BTV})
      The code $C_{\mathcal{H}_q}$ has length $n=(q^2-1)q,$ dimension $k = (q-1)q,$ and minimum distance $$d\geq (q+1)(q^2-3q+3) = q^3 - 2q^2 +3.$$
    \end{thm}

    Applying the viewpoint of \cite{HMM}, we are able to tighten this bound. 

    \begin{proposition}
    The code $C_{\mathcal{H}_q}$ has minimum distance $d$ satisfying
    \[d\geq q^3-2q^2+q+2.\]
      \begin{proof}
        First, we note that we may consider the Hermitian curve given as a fiber product as described above.
        Then $B_{\mathcal{H}_q}$ is the set of all points of $\X=\HH_q(\F_{q^2})$ lying above points of $\Y=\PP^1_u$ that split completely in the extension $\F_{q^2}(\X)/\F_{q^2}(\Y)$.
        We obtain $V_{\mathcal{H}_q}$ by letting $D$ be the zero divisor, so $l=0$.
        Applying Theorem \ref{thm:construction}, we find that the minimum distance $d$ of $C_{\mathcal{H}_q}$ is in fact bounded by 
          \[d\geq q^3-2q^2+q+2.\]
      \end{proof}
    \end{proposition}

    In Theorem \ref{thm:HermitianLRC2Dist}, we calculate exactly the minimum distance for this code.
  \end{example}

  \begin{example}{\textbf{LRC(2)s on the Fiber Product of two Hermitian Curves}}\label{ex:HermitianProduct}
    One can take the fiber product of any two curves with appropriate maps to the same base curve.
    As a simple example, we take a fiber product of two Hermitian curves.
    For $q$ a prime power, consider the Hermitian curves
      \[\HH_{q,1}\colon y_0^q+y_0=y_1^{q+1},\quad\text{and}\quad \HH_{q,2}\colon y_2^q+y_2=y_0^{q+1}. \]
    From each of these there is a projection to $\PP^1$ via the $y_0$ coordinate, $\PP^1_{y_0}$.
    Using these projections, we construct the fiber product $\mathcal{X}$.
    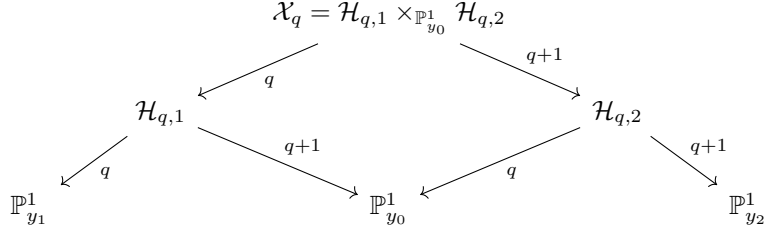
\begin{figure}
      \begin{center}
        \begin{tikzcd}
          &&\mathcal{X}_q=\HH_{q,1}\times_{\mathbb{P}^1_{y_0}} \HH_{q,2}\arrow{dr}{q+1}\arrow{dl}{q}& &\\
          &\HH_{q,1}\arrow{dr}{q+1}\arrow{dl}{q}&& \HH_{q,2}\arrow{dl}{q}\arrow{dr}{q+1}&\\
       \mathbb{P}^1_{y_1} &  &\mathbb{P}^1_{y_0} & & \mathbb{P}^1_{y_2}
        \end{tikzcd}
      \end{center}
      \caption{The Fiber product of two Hermitian curves.}
      \label{fig:FP of Hq}
      \end{figure}
    Intuitively, the affine part of this fiber product is pairs of points $(P_1,P_2)$ on the Hermitian curve $y_1^{q+1}=y_2^q+y_2$ satisfying $y_2(P_1)=y_1(P_2)$.
    Moreover
      \[\#\mathcal{X}_q(\mathbb{F}_{q^2})=q^4+1.\]
There is a single point at infinity for each Hermitian curve, so there is a single point at infinity for $\mathcal{X}_q$, which is totally ramified with ramification index $q(q+1)$ in the extension $\X_q/\mathbb{P}^1_{y_0}$.


    Also, defining $\Omega=\{\alpha \in \mathbb{F}_{q^2}:\alpha^q+\alpha=0\}$, we see that there are $q^2-q$ points of $\X_q(\mathbb{F}_{q^2})$ with $y_0$-coordinate $\alpha \not\in \Omega$ that are split completely in the extension $\mathcal{X}_q/\mathbb{P}^1_{y_0}$ (so have ramification index equal to $1$), and another $q$ points with $y_0$-coordinate $\alpha \in \Omega$ that ramify, but not completely; they have ramification index $q+1$.

    Since all ramification is tame in the extension $\X_q/\HH_{q,2}$, and $g(\HH_{q,2})=q(q-1)/2$, we can compute the arithmetic genus of the fiber product using Riemann-Hurwitz formula to get $g(\mathcal{X}_q)=q^3-q.$


    Following the construction from \cite{HMM}, we now present a code with two recovery sets by the evaluation of the splitting points on $\mathcal{X}_q(\mathbb{F}_{q^2})$.
    Let $B$ the set of $q^4-q^2$ points of $\mathcal{X}_q(\mathbb{F}_{q^2})$ that are above the $q^2-q$ points with $y_0$-coordinate $\alpha \in \mathbb{F}_{q^2}\setminus \Omega$ that split completely:
    $$B=\{(\alpha,\beta_1,\beta_2):\alpha \in \mathbb{F}_{q^2}\setminus \Omega, \, \beta_1^{q+1}=\alpha^q+\alpha \text{ and }\alpha^{q+1}=\beta_2^q+\beta_2\}.$$
    Let $P_\alpha=(\alpha,\beta_1,\beta_2)\in B$.
    Then  the sets $B_{\alpha}^j=\{(\alpha,y_1,y_2)\in B: y_k=\beta_k \ \forall k\neq j\}$ for $j=1,2$ are recovery sets for the position corresponding to $P_\alpha$.  We have $|B_{\alpha}^1|=q-1$ and $|B_{\alpha}^2|=q$.

    We define $$V=\textrm{Span}\{y_0^i\,y_1^j\,y_2^k:0\leq i \leq l, 0 \leq j \leq q-2, 0\leq k \leq q-1\},$$ with $l \leq\frac{q^4-2q^3+3q+1}{q(q+1)}$. 

     \begin{theorem} \label{thm:THC Code}
      For the fiber product $\mathcal{X}_q$ with $B$, $l$ and $V$ defined as above, the evaluation code $C(V,B)$ is a locally recoverable $[n,k,d]$-code over $\FF_{q^2}$ with availability $2$
      and locality $\left(q-1,q\right)$ where
      \begin{align*}
        n & = q^2(q^2-1), \\
        k &= (l+1)(q-1)q, \text{ and}\\
        d & \geq n-l q (q+1)-(q-1)q^2-(q-2)(q+1)^2.\\
      \end{align*}
    \end{theorem}
        In Corollary \ref{thm:MinimumDistanceTHC}, we calculate the minimum distance for this code.
  \end{example}


\begin{remark}\label{ex:HermitianDistance} Here we give a concrete example of the preceding construction of $C_{\X_3,3}$.
      Let $q=3$ and $\mathbb{F}_9=\mathbb{F}_3(a)$, where $a^2+2a+2=0$, be the finite field with 9 elements.
      Let us consider the situation of Example \ref{ex:HermitianProduct}, in which we have the fiber product of two Hermitian curves over $\F_9$:
      \[\mathcal{Y}_1\colon y_0^3+y_0=y_1^{4},\quad\text{and}\quad \mathcal{Y}_2\colon y_2^3+y_2=y_0^4, \]
      along with the fiber product $\mathcal{X}_3=\mathcal{Y}_1\times_{\mathbb{P}^1_{y_0}}\mathcal{Y}_2$.
      In this case, 
      $\Omega=\{\alpha\in \mathbb{F}_9:\alpha^3+\alpha=0\}=\{0, a+1,2a+2\}$, and we have 6 points on $\PP^1_{y_0}$ with first coordinate outside $\Omega$ that split completely in $\X_3$.
 The maximum $l$ that can be chosen to get a non-trivial bound for $d$ is $l=3$. Using this $l$ in Theorem \ref{thm:THC Code}, we get a LRC(2) of length $72$, dimension $k=24$ and minimum distance $d\geq 2$ over $\FF_9$ (an upper bound for the minimum distance is $35$, see Section \ref{SectionBounds}).
 
      
\end{remark} 

  \begin{example}{\textbf{Artin--Schreier Fiber Product and LRC($t$)}}\label{ex:ASExample}
    In \cite{HMM} the authors use a fiber product curve construction from van der Geer and van der Vlugt \cite{vdGvdVHowT0} to create codes with availability $t$ for arbitrary $t$.  Since we continue this example, we review the construction here.

    The simplest of the van der Geer and van der Vlugt constructions is given in \cite[Section 3, Method I]{vdGvdVHowT0}. 
    Let $p$ be prime, $h$ a natural number, and $q=p^{h}$. 
    Let $\{a_1,a_2, \dots, a_{h}\}$ generate $\ker(\Tr_{\mathbb{F}_{q^2}/\mathbb{F}_{q}})$ over $\mathbb{F}_p$. 
    Then the curves \[\mathcal{Y}_{i}\colon y_i^p-y_i=a_iy_0^{q+1}\] each have genus $\frac{1}{2}(p-1)q$ and have $pq^2+1$ points over $\mathbb{F}_{q^2}$, with one point, $\infty_{\mathcal{Y}_{i}}$, at infinity.

    Let $t$ be an integer with $1\leq t \leq h$ and let $\mathcal{Y}=\PP^1_{y_0}$.
    Then consider the natural map $h_i\colon \mathcal{Y}_{i}\rightarrow \mathcal{Y}$ given by projection onto the $y_0$ coordinate, where $\infty_{\mathcal{Y}}$ represents the point at infinity on the projective line $\mathbb{P}^1_{y_0}$ and $\infty_{\mathcal{Y}_{i}}\mapsto \infty_{\mathcal{Y}}$.
    These are all degree-$p$ Artin--Schreier covers of $\mathcal{Y}$, fully ramified above $\infty_\mathcal{Y}$.

    Define $\mathcal{X}=\mathcal{A}_{q,t}$ to be the fiber product of these curves $\mathcal{Y}_{i}$ over $\mathcal{Y}$; i.e.,
      \[\mathcal{A}_{q,t}=\mathcal{Y}_{1}\times_{\mathcal{Y}} \mathcal{Y}_{2}\times_{\mathcal{Y}} \dots 
\times_{\mathcal{Y}} \mathcal{Y}_{t}.\]
    The corresponding maps $g_i\colon \mathcal{A}_{q,t}\rightarrow \mathcal{Y}_{i}$ are degree $p^{t-1}$, ramified only above $\infty_{\mathcal{Y}_{i}}$.
    Let $\infty_{\mathcal{A}_{q,t}}$ be the single point above $\infty_{\mathcal{Y}}$ on $\mathcal{A}_{q,t}$.

    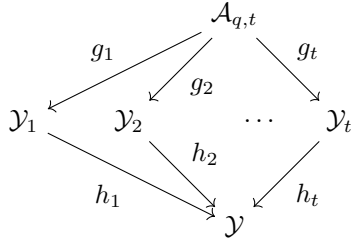
\begin{figure}[ht]
      \begin{center}
        \begin{tikzpicture}[node distance = 2cm, auto]
          \node (Y) {$\mathcal{Y}$};
          \node (Y1) [node distance=1.4cm, left of=Y, above of=Y] {$\mathcal{Y}_{2}$};
          \node (Ys) [node distance=1.4cm, right of=Y, above of=Y] {$\mathcal{Y}_{t}$};
          \node (X) [node distance=1.4cm, left of=Ys, above of=Ys] {$\mathcal{A}_{q,t}$};
          \node (Y2) [node distance=1.4cm, left of=Y1, left of=Ys, below of=X] {$\mathcal{Y}_{1}$};
          \draw[->] (Y1) to node {$h_2$} (Y);
          \draw[->] (Y2) to node [swap] {$h_1$} (Y);
          \draw[->] (Ys) to node {$h_t$} (Y);
          \draw[->] (X) to node [swap] {$g_1$} (Y2);
          \draw[->] (X) to node {$g_2$} (Y1);
          \draw[->] (X) to node {$g_t$} (Ys);
          \path (Y1) -- node[auto=false]{\ \ \ \ \ \ \ldots } (Ys);
        \end{tikzpicture}
      \end{center}
        \caption{The $t$-fold fiber product of Artin-Schreier curves, denoted $\mathcal{A}_{q,t}$.}
        \label{LRCt_fig}
    \end{figure}

    As shown in \cite[Theorem 3.1]{vdGvdVHowT0}, the curve $\mathcal{A}_{q,t}$ has genus $\frac{1}{2}(p^t-1)q$ and $|\mathcal{A}_{q,t}(\mathbb{F}_{q^2})|=p^tq^2+1$, making $\mathcal{A}_{q,t}$ maximal over $\mathbb{F}_{q^2}$.

    Note that the curve $\mathcal{A}_{q,t}$ is naturally a subvariety of $(\PP^2)^t$.  It embeds in $\PP^{t+1}$, however, by the map $\nu\colon\mathcal{A}_{q,t}\rightarrow\PP^{t+1}$ defined on affine points of $\mathcal{A}_{q,t}$ via
      \[\nu ((y_0,y_1),(y_0,y_2), \dots , (y_0,y_t))=(y_0,y_1,y_2,\dots, y_t).\] From here, we identify $\mathcal{A}_{q,t}$ with its image in $\PP^{t+1}$.
    The affine points of $\mathcal{A}_{q,t}$ are given by 
      \begin{equation}\label{eq:ASpoints}
        B=\{(y_0,y_1,y_2,\dots, y_t)\in (\mathbb{F}_{q^2})^{t+1}: y_i^p-y_i=a_iy_0^{q+1}\textrm{ for all } 1\leq i\leq t\}.
      \end{equation}
    For $1\leq i\leq t$, the functions $g_i\colon\mathcal{A}_{q,t}\rightarrow \mathcal{Y}_{i}$ are given by 
      \[g_i(y_0,y_1,y_2,\dots, y_t)=(y_0,y_i)\]
    and the functions $\tilde{g}_i\colon\mathcal{A}_{q,t}\rightarrow \tilde{\mathcal{Y}}_{i}$ are given by 
      \[\tilde{g}_i(y_0,y_1,y_2,\dots, y_t)=(y_0,y_1,y_2,\dots, y_{i-1},y_{i+1},\dots, y_t).\]
    For each $i$, the map $\tilde{g_i}$ has degree $p$.
    For $1\leq i\leq t$, the function $y_i$ has degree $d_{y_i}=q+1$, since for each $\alpha, \beta\in\mathbb{F}_{q^2}$ with $\beta\neq 0$ and $\alpha^p+\alpha = a_i\beta^{q+1}$, there are $q+1$ points 
    $Q_j=(\zeta^k\beta,\alpha)\in\mathcal{Y}_{i}(\mathbb{F}_{q^2})$, where $\zeta^{q+1}=1$ and $1\leq k\leq q+1$.

    \begin{remark}
      As observed in \cite{HMM}, when $t=h$, we have that $\mathcal{A}_{q,t}\cong\mathcal{H}_{q}$.
    \end{remark}

    Applying the construction from \cite{HMM}, we can construct codes defined over $\mathbb{F}_{q^2}$ with many recovery sets.  
    Let $P_i=(\alpha,\beta_1,\beta_2,\dots, \beta_t)\in B$.
    Then $B_{i,j}$, the $j$-th recovery set for the position corresponding to $P_i$, is the set of positions corresponding to the points in $\{(\alpha,y_1,y_2,\dots, y_t)\in B:~ y_k=\beta_k~ \forall ~k\neq j\}.$  We then have $|B_{i,j}|=p$.
    On points corresponding to the positions in $B_{i,j}$, any function in $V$ varies as a polynomial in $y_j$ of degree at most $(p-2)$ and can therefore be interpolated by knowing its values on any $p-1$ points.

    Given $h, t$ as above, choose $l\leq \left(q^2-\frac{t(p-2)(q+1)p^{t-1}+1}{p^t}\right)$ to ensure the evaluation map is injective. Note that the evaluation map may be injective for larger values of $l$ but that the given lower bound ensures that $d\geq 1$ in the Theorem below.  
    Let $D=l\infty_{\Y}$.
    Then $\mathcal{L}(D)$ is the set of polynomials in $y_0$ of degree at most $l$, a vector space of dimension $m=l+1$.

    \begin{theorem}[\cite{HMM}] \label{thm:ASCode}
      Given $\X=\mathcal{A}_{q,t}$ the fiber product of the specified Artin--Schreier curves, with $B$ and $l$ as above, let $D=l\infty_{\Y}$, and $V$ as defined in Theorem \ref{thm:construction}.  We define $C_{\mathcal{A}_{q,t},l}=C(V,B)$.
      Then $C_{\mathcal{A}_{q,t},l}$ is a locally recoverable $[n,k,d]$-code over $\FF_{q^2}$ with availability $t$
      and locality $\left(p-1,p-1,\dots,p-1 \right)$ where
      \begin{align*}
        n & = p^tq^2, \\
        k &= (l+1)(p-1)^t, \text{ and}\\
        d & \geq n-lp^t-t(p-2)(q+1)p^{t-1}. \\
      \end{align*}
    \end{theorem}

    In Theorem \ref{thm:ASDistance}, we compute the exact minimum distance of the code here for many values of $l$.

  \end{example}

\section{Computing Minimum Distances}\label{sec:mindist}

  A standard technique for determining the minimum distance of an evaluation code $C(V,B)$ is to first bound the minimum distance below using a geometric argument, then find an element in the space of functions $V$ that vanishes at the maximum number of points, all of which are contained in the evaluation set $B$.
  Using the bounds from \cite{HMM} and this technique, it is possible to find the exact minimum distance of some interesting codes from \cite{BTV} and \cite{HMM}.

  As a warm-up, we find the exact minimum distance of two LRCs on the Hermitian curve described in \cite{BTV}.
  The first arises from a simpler rational map construction.
  Let $C$ be the code with locality $q-1$ described in Proposition 4.1 of \cite{BTV}, i.e., the evaluation code $C(V,B)$ where $\HH_q$ is the Hermitian curve defined by $y^q+y=x^{q+1}$, $B$ is the set of $q^3$ affine points in $\HH_q(\F_{q^2})$, and $V$ is the vector space of functions generated by $\{x^iy^j\colon 0\leq i \leq l, 0\leq j\leq q-2\}$ for some fixed $l\in\mathbb{N}$.
  Note that the recovery set for the position corresponding $P\in B$ is the set of $q-1$ points 
  \[\{Q \in B: x(Q)=x(P), Q\neq P\}.\]

  \begin{thm}\label{thm:HermitianLRC2Dist} 
    When $l\leq q^2-q-2$, the code $C$ has minimum distance $$d = n-lq - (q-2)(q+1).$$
  \end{thm}
  \begin{proof}
   In \cite{BTV}, the authors prove that $n-lq-(q-2)(q+1)$ is a lower bound on the minimum distance of $C$.
    Suppose $l\leq q^2-q-2$. 
    Considering the extension of fields $\FF_{q^2}$ to $\FF_{q}$, let $\varphi_1$ be the field trace map given by
    $\varphi_1(x) = x^{q}+x$ and let $\varphi_2$ be the norm map given by $\varphi_2(x) = x^{q+1}$.
    Since $\varphi_1$ is the trace map, which is degree $q$ onto $\mathbb{F}_q$, we can write 
      \[\varphi_1^{-1}(1) = \{\gamma_1, \dots, \gamma_{q}\}.\]
    Since $\varphi_2$ is the norm map and so is degree $q+1$ onto $\mathbb{F}_q^{\times}$, we can write
      \[\FF_{q^2}\backslash(\{0\}\cup \varphi_2^{-1}(1)) = \{\beta_1,\dots,\beta_{q^2 - q - 2}\}.\]
    Define $f\in V$ by 
      \[f(x,y) = \prod_{j=1}^l(x-\beta_j)\prod_{i=1}^{q-2}(y-\gamma_i).\]
  We see that $f$ has at most $lq + (q-2)(q+1)$ zeros.
    To show that $f$ has exactly that many zeros, we must show that no evaluation point of $C$ is sent to zero by more than one factor of $f$. 
    Suppose $f(\beta_j,\gamma_i) = 0$.
  Then $\gamma_i^q+\gamma_i=1$, but by design $\beta_j^{q+1}\neq 1,$ hence $(\beta_j,\gamma_i)\not\in \mathcal{H}_q(\FF_{q^2}).$ Thus no evaluation point can be sent to zero by multiple factors, so $f$ has exactly $lq+(q-2)(q+1)$ zeros and $ev_B(f)$ has weight $n-lq-(q-2)(q+1),$ and the minimum distance is as given.

\end{proof}

  Recall that $C_{\mathcal{H}_q}$ is the LRC(2) on $\HH_q$ defined in \cite{BTV}.
  We now determine the exact minimum distance of the code.
  
  \begin{thm}\label{thm:BTVHqmindist}
    The code $C_{\mathcal{H}_q}$ has minimum distance $d=q^3-2q^2+q+2$.
    \begin{proof}

      Now let $\alpha_1,\alpha_2\in \FF_{q}\backslash\{0\}$ such that $\alpha_1\neq \alpha_2$.
      Then let $E = \{a\in \FF_{q^2} | a^{q+1} = \alpha_1\}$ and $F = \{a\in \FF_{q^2} | a^{q}+a = \alpha_2\}$.
      Because these come from the trace and norm respectively, we can write $E = \{\beta_1, \dots, \beta_{q+1}\}$
      and $F=\{\gamma_1, \dots, \gamma_{q}\}$.
      Then let $f\in V_{\mathcal{H}_q}$ be defined 
        \[f = \prod_{i=1}^{q-1}(x-\beta_i)\prod_{j=1}^{q-2}(y-\gamma_j).\]
      The function $f$ has exactly $(q-1)(q) + (q-2)(q+1)$ zeros by the same argument as in the previous construction.
      Thus $d\leq q^3-q-(q-1)(q) - (q-2)(q+1)=q^3-2q^2+q+2.$
    \end{proof}
  \end{thm}

  Next, we determine the exact minimum distance for many codes from the fiber product of Artin--Schreier curves constructed in Theorem~\ref{thm:ASCode}.

  \begin{thm}\label{thm:ASDistance}
    Let $p$ be a prime and $q=p^h$ a prime power. For a fixed $l\in\mathbb{Z}$ with $0\leq l \leq q^2-tq-t-1$, let $C_{A_{q,t},l}$ be the LRC$(t)$ of Theorem~\ref{thm:ASCode}, constructed using the fiber product of $t$ Artin--Schreier curves.
    Then $C_{A_{q,t},l}$ has minimum distance 
    $$d=p^tq^2 - lp^t - t(p-2)(q+1)p^{t-1}.$$
    \begin{proof}
      Recall that the curves that we use to produce the fiber product $\X$ are of the form
        \[\Y_i\colon y_i^p-y_i=a_iy_0^{q+1},\]
      where $ \langle a_1,a_2,\dots,a_h\rangle_{\F_p}=\ker(\Tr_{\F_{q^2}/\F_{q}})$.
      From Theorem~\ref{thm:ASCode}, we have a lower bound for minimum distance,
        $d\geq p^tq^2 - lp^t - t(p-2)(q+1)p^{t-1}$.

      Let $\varphi_1\colon \FF_{q^2} \to \FF_{q^2}$ be defined by
        $\varphi_1(x) = x^{p}-x$
      and $\varphi_2\colon \FF_{q^2}\to \FF_{q}$ be defined by the norm map
        $\varphi_2(x) = x^{q+1}$.
      Choose values
        \[F_0 = \{\beta\in\F_{q^2}^\times : \varphi_2(a_i\beta)\neq 1\ \forall\ 1\leq i\leq t\}\]
      and $F_i = \varphi_1^{-1}(a_i^{-q})$.
      Since $\varphi_2$ is the norm map, $|\varphi_2^{-1}(1)| = q+1$, so $|F_0|=q^2-t(q+1)-1$.
      Thus we can write 
        \[F_0 = \{\beta_1,\dots,\beta_{q^2-tq-t-1}\}.\]

      By choice of $a_i$, we have $\Tr_{\F_{q^2}/\F_q}(a_i) = 0,$ and since $\F_{q^2}/\F_q$ is a degree 2 extension, $\Tr_{\F_{q^2}/\F_q}(a_i^{-1})  = 0$ too.
      Then since we are working in characteristic $p$ and the trace map can be factored,
          \begin{align*}
            \Tr_{\FF_{q^2}/\FF_p}(a_i^{-q}) & = \Tr_{\FF_{q}/\FF_p}(\Tr_{\FF_{q^2}/\FF_q}(a_i^{-q}))\\
            &= \Tr_{\FF_{q}/\FF_p}(\Tr_{\FF_{q^2}/\FF_q}(a_i^{-1})^q)\\
            &= 0.
          \end{align*}
      By the additive version of Hilbert's Theorem 90, 
      $ \varphi_1^{-1}(a_i^{-q})$ is nonempty. 
      Notice that $\varphi_1$ is separable of degree $p$, so $|F_i|=|\varphi_1^{-1}(a_i^{-q})|$ must, in fact, equal $p$.
      We then write
        \[F_i =\{\gamma_{i,1},\dots,\gamma_{i,p}\}.\]
      Define the map $f\colon\X\to \FF_{q^2}\in V$ by
        \[f(y_0,y_1,\dots,y_t) = \prod_{i=1}^l(y_0-\beta_i)\prod_{j=1}^{t}\prod_{k=1}^{p-2}(y_j-\gamma_{j,k}).\]

      Recall that $B$, as in equation \eqref{eq:ASpoints}, is the evaluation set for polynomials in $V$ and $|B|=n=p^tq^2$ is the length of $C_{A_{q,t},l}$. Certainly $n-\wt(\ev_B(f))$ is at most $l p^t + t(p-2)(q+1)(p^{h-1})$.
      We wish to show that these values are equal by showing that $f$ has exactly this many zeros in $B$ by showing that no points in $B$ have $y_i$ component in $F_i$ and $y_j$ component in $F_j$ for all $i\neq j$.
      Toward that end, suppose that $(\beta,\gamma_1,\dots,\gamma_t)\in B$ such that $f(\beta,\gamma_1,\dots,\gamma_t)=0$.

      Assume that $\gamma_i\in F_j$ for some $i\in \{1,\dots,t\}$. 
      Then $a_i^{-q} = \phi_2(\gamma_i) = \gamma_i^p-\gamma_i = a_i\beta^{q+1}$, so $a_i^{-q-1}=\beta^{q+1}$, by definition of $\X$.
      This equation has at most $q+1$ solutions, all of which are in the set $\{a_i^{-1}\varphi_1^{-1}(1)\}$.
      Since none of these are in $F_0$, we have that $\beta\notin F_0$.

      Now suppose that for some $j\in \{1,\dots,t\}$ we have $\gamma_j\in F_j$.
      Then $a_j^{-q} = \varphi_2(\gamma_j) = \gamma_j^p+\gamma_j = a_j\beta^{q+1}$.
      Thus we have $a_i^{q+1} = a_j^{q+1}$.
      Recall that $a_i$ is in the kernel of the trace map, so $a_i^{q}+a_i = 0$ and so $a_i^{q+1}= -a_i^2$.
      Similarly $a_j^{q+1} = -a_j^2$.
      Substituting these values gives us the equality $a_i^2=a_j^2$, so $a_i=a_j$ or $a_i=-a_j$,
      but $a_i$ and $a_j$ are elements of a basis for the kernel of $\Tr_{\F_{q^2}/\F_q}$ over $\F_p$, so $i=j$.

      Thus $f$ has exactly $l p^t + t(p-2)(q+1)(p^{t-1})$ zeros in $B$, so the code has the desired minimum distance,
        \[d=p^t{q^2} - lp^t - t(p-2)(q+1)(p^{t-1}).\]
    \end{proof}

  \end{thm}

\subsection{A Condition for Exact Minimum Distance.}
  More generally, we may summarize the situation in which this technique will give the exact minimum distance of codes from the construction in Theorem \ref{thm:construction}.
  \begin{theorem}\label{thm:MinimumDistance}
    Let $C(V,B)$ be a locally recoverable code constructed as in Section \ref{sec:CodeConstruction}, where $V$ has basis given by \eqref{eqn:Vbasis}, $B$ is the evaluation set as in \eqref{eqn:Bdef}, and $y_0^j\in\mathcal{L}(D)$ for $0\leq j \leq l$. If it is possible to find sets $F_0,F_1,\dots,F_t\subseteq\F_q$ such that 
    \begin{enumerate}[label=$(\arabic*)$]
      \item $F_i\subseteq y_i(B)$ for all $i=0,\dots,t$,
      \item $|F_0|=l$,
      \item $|F_i|\geq d_{h_i}-2$ for all $i=1,\dots,t$,
      \item\label{assumption:coordmatch} for all $i\neq j$ with $0\leq i,j\leq t$ there is no $P\in \X(\overline{\F_q})$ with $y_i$-coordinate in $F_i$ and $y_j$-coordinate in $F_j$, and 
      \item\label{assumption:notramified} for all $i$ with $0\leq i\leq t$, the projection $y_i\colon\X\to \PP_{y_i}^1$ is not ramified over any point $P\in\PP^1_{y_i}$ with $y_i(P)\in F_i,$
    \end{enumerate}
    then the code $C(V,B)$ has minimum distance
      \begin{equation}\label{eq:mindist}d=n-l
        d_{y_0}-\sum_{i=1}^t(d_{h_i}-2)d_{y_i},
      \end{equation}
    where $n=|B|$ is the length of the code.
  \end{theorem}

  \begin{remark}
   If $\Y=\PP^1_{y_0}$, $h_i:\Y_i\to \Y$ given by projection onto $y_0$, and $D=l\infty_{\Y}$, as is the case in all examples in this paper, we have that $y_0^j\in\mathcal{L}(D)$ for $0\leq j \leq l$ and $y_0:\X \to\Y$ is an unramified map above $S$ from the code construction.
  \end{remark}
  \begin{proof}
    By Theorem~\ref{thm:construction}, the right hand side of \eqref{eq:mindist} is a lower bound on the minimum distance of such a code.

    Now, label the elements of sets $F_0,F_1,\dots,F_t$ as
      \[F_0=\{\beta_1,\dots,\beta_l\}\text{ and }F_i=\{\gamma_{i,1},\dots,\gamma_{i,|F_i|}\}.\]
    Then we can define the polynomial in $V$,
      \[f=\prod_{j=1}^l(y_0-\beta_j)\prod_{i=1}^t\prod_{k=1}^{d_{h_i}-2}(y_i-\gamma_{i,k}).\]
    Since the points in $B$ are fully split in the extension $\F_q(\X)/\F_q(\Y)$ and assumption \ref{assumption:notramified} we have that $|y_0^{-1}(\beta_j)|=d_{y_0}$ for all $j$ and $|y_i^{-1}(\gamma_{i,k})|=d_{y_i}$ for all $i,k$.
    By assumption \ref{assumption:coordmatch}, we know that $f$ must have exactly
      \[ld_{y_0}+\sum_{i=1}^t(d_{h_i}-2)d_{y_i}\]
    zeros, so the code has minimum distance
      \[d=n-\left(ld_{y_0}+\sum_{i=1}^t(d_{h_i}-2)d_{y_i}\right).\]
  \end{proof}
  \subsubsection{Examples of Applying Theorem \ref{thm:MinimumDistance}}
Here, we give two extremely concrete examples to illustrate the application of this general condition.

    \begin{example}\label{ex:ASExample2}
      Let $p=3$ and $t=h=2$, so $q=p^h=9$ and we work over $\F_{p^{2h}}=\F_{81}$.  Let $b$ be a non-trivial fifth root of unity for which $\F_{p^{2h}}= \F_p(b)$. 
      Let $a_1 = b^2 + b + 2$ and $a_2 = b^3 + b + 2$ be generators of
      $\ker(\Tr_{\F_{p^{2h}}/\F_{p^2}})=\{x\in\F_{p^{2h}} : a^9+a=0\}$. Then we have explicit curves
        \[
        \Y_1\colon y_i^3-y_i=(b^2 + b + 2)y_0^{3^2+1}\quad\text{and}\quad
        \Y_2\colon y_i^3-y_i=(b^3 + b + 2)y_0^{3^2+1}.
        \]
      Each of these curves has a projection onto $\PP^1$ via their $y_0$-coordinate, which we will denote $h_1$ and $h_2$. Consider their fiber product
        \[\mathcal{A}_{9,2}=\Y_1\times_{\PP^1}\Y_2, \]
      which is a genus 36 curve. Each of the maps $g_i\colon\X\to\Y_i$ are degree 3 and ramified only above the point at infinity. We can realize the 729 affine points of $\mathcal{A}_{9,2}(\F_{p^{2h}})$ to be the set
        \[\mathcal{P}=\left\{(y_0,y_1,y_2)\in\A_{\F_{p^{2h}}}^3 : y_i^3+y_i=a_i y_0^{3^2+1}\ \forall\ i\right\}.\]

      In order to satisfy the hypothesis of Theorem \ref{thm:MinimumDistance}, it will suffice to find sets $F_0$, $F_1$, and $F_2$ with $|F_1|=|F_2|=1$ and $|F_0|=l$, where $D=l\infty_{\Y}$.
      We choose $F_1 = \{b^2+b\}$ and $F_2 = \{b^3+2b^2\}$, each of which will eliminate $10$ possible $y_0$ values from entry into $F_0$, as there are 10 points in $\mathcal{P}$ with $y_1$-coordinate $b^2+b$ and 10 with $y_2$-coordinate $b^3+2b^2$.
      Thus $61$ values remain as possible elements of $F_0$ and hence we may use any value $0\leq l\leq 61-1$, to receive a code with the prescribed minimum distance, as in Theorem \ref{thm:MinimumDistance}.
    \end{example}

    \begin{example}\label{ex:HermitianDistance2}
      If we try to apply Theorem \ref{thm:MinimumDistance} to get the exact minimum distance for the code $C_{\X_3,3}$ over $\mathbb{F}_9$ constructed in Remark \ref{ex:HermitianDistance} we see that the sets $F_0$, $F_1$ and $F_2$ cannot be built, so a minimum weight codeword cannot be constructed by this method and the miniumum distance cannot be determined by the theorem.
      Instead, let us consider the situation of Example \ref{ex:HermitianProduct} over the finite field $F_{4^2}$, i.e. the fiber product $\mathcal{X}_{4}=\mathcal{Y}_1\times_{\mathbb{P}^1_{y_0}}\mathcal{Y}_2$, where we define $\Y_1$ and $\Y_2$ to be copies of the Hermitian curve $\HH_4$ with equations given by:
      \[\mathcal{Y}_1\colon y_0^4+y_0=y_1^{5},\quad\text{and}\quad \mathcal{Y}_2\colon y_2^4+y_2=y_0^5. \]

      In this case, if $a \in \F_{16}$ is such that $a^4+a=1$, then the finite ramified points in $\X_{4}$ have first coordinate in $\Omega=\{\alpha\in \mathbb{F}_{16}:\alpha^4+\alpha=0\}=\{0, a^2+a,a^2+a+1,1\}$, and we have 12 points on $\PP^1_{y_0}$ with first coordinate outside $\Omega$ that split completely in $\X_{4}$.

 The maximum $l$ that can be chosen to get a non-trivial bound for $d$ is $l=6$. But we can not build a set $F_0$ with $6$ elements satisfying the hypothesis of Theorem \ref{thm:THC Code}. So we will build one using $l=4$. Let $B$ be the set of 240 evaluation points in $\mathcal{X}_4(\mathbb{F}_{16})$ such that $y_0(B)=\FF_{16}\setminus \Omega$ is the set of $y_0$-coordinates.
       Defining $F_0=\{a^3+a+1, a^3+a^2+1, a^3+a^2+a, a^3+1\}$, $F_1=\{a^3, a^3+a^2, a^3+a, a^3+a^2+a+1, 1\}$ and $F_2=\{a, a^2, a+1, a^2+1\}$, we can see that the hypothesis of Theorem \ref{thm:MinimumDistance} hold and therefore $C_{\X_4,4}$, i.e. the evaluation code of functions from 
      $$V=\Span\{y_0^iy_1^{e_1}y_2^{e_2}|i=0,\ldots, 4;e_1=0,\ldots, 3;e_2=0,1, 2\},$$ evaluated at points in $B$, is an $[240,60,62]$-locally recoverable code with availability 2. Every coordinate in a codeword can be recovered using two possible sets: one with 3 elements and another with 4 elements, giving a locality of $(3,4)$. 
    \end{example}

The situation of the previous example can be generalized to compute the exact minimum distance for the code $C_{\X_q,l}$ as follows.

\begin{corollary}\label{thm:MinimumDistanceTHC} Let $q>3$ and $\mu \in \FF_{q}^{\times}$ such that 
$\mu\neq \alpha^{q+1}$ for all $\alpha \in \FF_{q^2}$ such that $\alpha^{q+1}=\alpha^q+\alpha$. 
For $0\leq l \leq q$, the code $C_{\X_q,l}$ of Theorem \ref{thm:THC Code} over $\FF_{q^2}$ has minimum distance $$d=n-l q (q+1)-(q-1)q^2-(q-2)(q+1)^2.$$
\end{corollary}

\begin{proof}

Let \begin{align*}
F_0&\subseteq\{x \in \FF_{q^2}: x^{q+1}=x^q+x\}\setminus\{0\},\\
F_1&=\{x \in \FF_{q^2}: x^{q+1}=\mu\}\\
\intertext{and}
F_2&=\{x \in \FF_{q^2}:x^q+x=\mu\}.\\
\end{align*} By construction, $|F_0|=l$ with $0\leq l \leq q$, $|F_1|=q+1$ and $|F_2|=q$. Also, if $P=(\alpha, \beta, \gamma)\in B$ then $\alpha^q+\alpha \neq 0$, $\beta^{q+1}=\alpha^q+\alpha$ and $\gamma^q+\gamma=\alpha^{q+1}$. Therefore
$F_i\subseteq y_i(B)$ for $i=0,1,2$. 
Moreover, if $P\in \X(\overline{\F_q})$ has $y_i$-coordinate in $F_i$ then its $y_j$-coordinate is not in $F_j$ for $j\neq i$. In fact, if $\alpha \in F_0$, then $\alpha^q+\alpha=\alpha^{q+1}$, so $\beta \in F_1$ yields to a contradiction since in this case $\mu=\beta^{q+1}=\alpha^q+\alpha=\alpha^{q+1}$, and the same happens if $\gamma \in F_2$.
A similar argument shows that the other two cases can not occur either.
 Therefore, Theorem \ref{thm:MinimumDistance} holds.
\end{proof}

\begin{remark}
Notice that for many examples, we can choose $\mu=1$. Actually this is the case, for example, for $q=4,7, 13, 16, 19, 25$. For $q=5, 9, 17$ we can use $\mu=2$ and for $q=11$, $\mu=5$ satisfies the required hypothesis. 
\end{remark}

  \subsection{A Combinatorial Condition for Exact Minimum Distance}
    
    We apply a very simple counting argument to show that the conditions of Theorem \ref{thm:MinimumDistance} hold when the evaluation set is large enough in relation to the map degrees and the base curve of the fiber product is $\Y=\PP^1_{y_0}$. Let $S=S_0$ be the set of points on $\Y$ lying below the points of $B$, and let $S_i$ be the set of points of $\mathcal{Y}_{i}$ lying below the points of $B$ for each $i$, $1\leq i\leq t$.
    As a non-infinite point of the projective line, each point of $S_0$ corresponds to a value $\alpha$ in $\mathbb{F}_q$.

    \begin{theorem}\label{thm:counting}
      Let $C(V,B)$ be a code constructed from a fiber product as in Theorem \ref{thm:construction}, where $\Y=\mathbb{P}^1_{y_0}$, and let $\eta_0=1$ and $\eta_i=\deg(h_i)$ for $1\leq i \leq t$.  
      Let $\psi_0=l$ and $\psi_i=\deg(y_i)$ for $1\leq i \leq t$, where here we consider the function $y_i:\Y_i\rightarrow\mathbb{P}^1_{y_i}$. Then the conditions of Theorem \ref{thm:MinimumDistance} above hold whenever 
        \[|S_i|\geq\sum_{i\neq j}(\eta_i-2)\psi_i \eta_j\psi_j\] for all $1\leq i\leq t$ and 
        \[|S_0|\geq\sum_{j=0}^t\eta_j\psi_j.\]

    \end{theorem}

    \begin{proof}
      For each $i$, $\leq i \leq t$, let $T_i\subseteq \mathbb{F}_q$ be the set of values of the $y_i$-coordinates of points in $S_i$. Note that $|S_0|=|T_0|$, and that $|S_i|=\eta_i|S_0|=\eta_i|T_0|$.  
      We will proceed by removing points from $S_i$ and values from $T_i$ as we build the sets $F_i$.  We will be successful in constructing the function in the proof of Theorem \ref{thm:MinimumDistance} if we construct all the sets $F_i$ without exhausting the sets $S_j$ and $T_j$ for any $j$. 

      First, let $F_0$ be any set of $l$ elements of $T_0$.  Remove these elements from $T_0$.  For each $i$, $1\leq i\leq t$, each of these $y_0$-values will be present in at most $\eta_i$ points in $S_i$, which will cover a total of at most $l\eta_i$ values of $y_i$. 
      Remove these values from $T_i$ for each $i$. 
      These values of $y_i$ will each appear in at most $\deg(y_i)=\psi_i$ points of $S_i$.  Remove these points from $S_i$.  This accounts for at most $l\eta_i\psi_i=\eta_0\psi_0\eta_i\psi_i$ points in $S_i$ for each $i$, $1\leq i\leq t$.  

      Beginning with $i=1$, let $F_1$ consist of any $\eta_1-2$ values of $y_1$ which appear as $y_1$ coordinates in $S_1$.  These values of $y_1$ will appear in at most $(\eta_1-2)\psi_1$ points in $S_1$, which will lie above at most $(\eta_1-2)\psi_1$ points in $S_0$.  Remove these points from $S_0$, and these $y_0$-values from $T_0$. Note that by design, these values will not have been previously removed from $T_0$.  For each of these values of $y_0$, there are at most $\eta_j$ points in $S_j$ with these $y_0$-values, meaning at most $\eta_j$ values of $y_j$ across these points.  Remove these points from $S_j$ and these values of $y_j$ from $T_j$ for all $j$, $1\leq j\leq t$, $j\neq i$.  There are a total of $\eta_j\psi_j$ points with these values of $y_j$ in $S_j$.  By assumption, the sets $S_j$ were all large enough that this must be possible. Repeat for all $i$, $2\leq i\leq t$, building sets $F_2,\dots, F_t$. This is possible as long as no set $T_0$, $S_0$, $T_i$ or $S_i$ is empty at any point in the process. By definition, the set $T_i$ must be non-empty as long as $S_i$ is non-empty.  At each stage, we remove $(\eta_i-2)\psi_i\eta_j\psi_j$ points from each $S_i$ for $i\geq 1$ and $(\eta_j-2)\psi_j$ points from $S_0$. Since $|S_0|\geq\sum_{j=0}^t\eta_j\psi_j$ and $|S_i|\geq\sum_{i\neq j}(\eta_i-2)\psi_i \eta_j\psi_j$, we always have enough points to do this.  Thus all sets $F_i$ can be constructed this way.

    \end{proof}

    \begin{example} 
      We know from Theorem \ref{thm:ASDistance} the exact minimum distance of $C_{\mathcal{A}_{q,t},l}$ for many values of $l$.  However, we apply Theorem \label{thm:counting} to find the exact minimum distance of codes obtained from the curves of Example \ref{ex:ASExample} over an extended base field.
      In particular, consider points on the base curve, factor curves, and fiber product defined over $\FF_{q^6}=\FF_{p^{6h}}$, where for simplicity we take $h=t$ to also be the number of factor curves $\Y_{a_i}$.
      Since the curve $\mathcal{A}_{q,t}$ is maximal over $\F_{q^2}$, it is also maximal over $\F_{q^6}$ (upon consideration of the $L$-function of the curve).
      Thus we can compute that $\mathcal{A}_{q,t}$ has $p^{6h}+p^{5h}-p^{4h}+1$ points over $\mathbb{F}_{p^{6h}}$.
      Since each curve $\Y_{a_i}$ is covered by the maximal curve $\mathcal{A}_{q,t}$, $\Y_{a_i}$ is also maximal and thus has $p^{6h}+p^{4h+1}-p^{4h}+1$ points over $\mathbb{F}_{p^{6h}}$.
      Note that each $\FF_{p^{6h}}$-point corresponds to a place of degree 1 in the function field $\FF_{p^{6h}}(\Y_{a_i})$.

      First we consider the lower Artin--Schreier extensions $\FF_{p^{6h}}(\Y_{a_i})/\FF_{p^{6h}}(y_0)$, corresponding to the maps $h_i$ from the curves $\Y_{a_i}$ to projective line by projection onto the $y_0$-coordinate.
      These Artin--Schreier extensions of the projective line are described completely in 
      \cite[3.7.8 and 6.4.1]{stichtenoth}.
      The extensions are Galois of degree $p$.
      Each degree-one place in $\FF_{p^{6h}}(\Y_{a_i})$ lies above a fully ramified or fully split place in $\FF_{p^{6h}}(y_0)$.
      The only ramified place in this extension is the unique place at infinity.
      Thus the $p^{6h}+p^{4h+1}-p^{4h}$ affine rational points of $\Y_{a_i}$ over $\FF_{p^{6h}}$ arise from $p^{6h-1}+p^{4h}-p^{4h-1}$ places in $\FF_{p^{6h}}(y_0)$ splitting completely.  

      Recall that the function field of the fiber product of curves is the compositum of the function fields of the curves.
      Extending \cite[Proposition 3.9.6]{stichtenoth}, we see that if a place of $\FF_{p^{6h}}(x)$ splits completely in each extension $\FF_{p^{6h}}(\Y_{a_i})/\FF_{p^{6h}}(y_0)$, then this place splits completely in the compositum extension $\FF_{p^{6h}}(\mathcal{A}_{q,t})/\FF_{p^{6h}}(y_0)$.
      Since all non-infinite degree-one places of $\FF_{p^{6h}}(\mathcal{A}_{q,t})$ must lie above non-infinite degree-one places of $\FF_{p^{6h}}(\Y_{i})$, we have that all the non-infinite degree-one places of $\FF_{p^{6h}}(\mathcal{A}_{q,t})$ lie above places of $\FF_{p^{6h}}(y_0)$ which split completely in the degree $p^h$ extension $\FF_{p^{6h}}(\mathcal{A}_{q,t})/\FF_{p^{6h}}(y_0)$.
      Since there are $p^{6h}+p^{5h}-p^{4h}$ non-infinite degree-one places of $\FF_{p^{6h}}(\mathcal{A}_{q,t})$, these lie above $p^{5h}+p^{4h}-p^{3h}$ non-infinite degree-one places of $\FF_{p^{6h}}(y_0)$ which split fully in all extensions. 

      Applying the construction from Section \ref{sec:CodeConstruction}, we may take the evaluation set $B$ to be the set of all affine points of $\mathcal{A}_{q,t}(\FF_{p^{6h}})$ and the divisor $D=l\infty_{\mathcal{Y}}$ for any $l$ with $l\leq p^{5h}+p^{4h}-p^{3h}-1$ for guaranteed positive minimum distance.
      By Theorem \ref{thm:construction}, we get a locally recoverable code with uniform locality $p-1$, availability $h$, length $n=p^{6h}+p^{5h}-p^{4h}$, dimension $l(p-1)^h$, and minimum distance $d \geq n - lp^{h} - h(p-1)p^{2h-2}(p^{h}+1)$.

      Let $S_0$ be the points of $\Y$ corresponding to these fully split places below $B$.
      Let $S_i$ be the points on $\Y_{a_i}$ lying above $S_0$ for each $i$.
      We then have that $|S_i|=p|S_0|=p^{5h+1}+p^{4h+1}-p^{3h+1}$.

      To apply Theorem \ref{thm:counting}, we note that $\eta_0=1$, $\psi_0=l$, $\eta_i=p$ and $\psi_i=p^h+1$.
      Then
        \[\sum_{i\neq j}(\eta_i-2)\psi_i \eta_j\psi_j\leq(t-1)(p^{2h+2}+2p^{h+2}+p^2)+l(p^{h+1}+p)\]
      for all $1\leq i\leq t$ and 
        \[\sum_{j=0}^t\eta_j\psi_j=(t-1)(p^{h+1}+p)+l.\]
      Note that there is a large range of values of $l$ for which the conditions of Theorem \ref{thm:counting} hold, and thus for which the bound on minimum distance given in Theorem \ref{thm:construction} is the true minimum distance.
    \end{example}

\section{Parameter Ranges and Comparison to Bounds}\label{SectionBounds}

  We find the following bounds on the parameters of LRC($t$)s in the literature.  
  \begin{itemize}
    \item For all codes, the Singleton bound:
    \begin{equation*}
    d\leq n-k+1.
\end{equation*}

    \item Tamo and Barg proven rate bound \cite{TamoBarg} (2014), for codes with uniform locality $r$:
    \begin{equation}\label{eqn:TamoBargBounds}
    d\leq n-\sum_{i=0}^t\left\lfloor\frac{k-1}{r^i}\right\rfloor, \hspace{.5in} R \leq \frac{1}{\prod_{j=1}^{t}\left(1+\frac{1}{jr}\right)}.
    \end{equation}

    \item Bhadane and Thangaraj \cite{BhadaneThangaraj} (2017), for codes with locality $(r_1, r_2,\dots, r_t)$, where $r_i\leq r_j$ for $i<j$:
    \begin{equation}\label{BTBound}
    d\leq n-k+1-\sum_{i=1}^{t}\left\lfloor \frac{k-1}{\prod_{j=1}^{i}r_j}\right\rfloor.
    \end{equation}



  \end{itemize}

 The proven rate bound in \eqref{eqn:TamoBargBounds} is known to be tight for $t=1$ but no constructions have realized this bound for $t\geq 2$.  Two constructions for general $r$ and $t$ should be mentioned here.  First, in \cite{TamoBarg}, Tamo and Barg consider a binary code which is the product of $t$ single-parity-check codes with $r$ message symbols each. This gives an LRC($t$) with locality $r$ for each recovery set for arbitrary $r$ and $t$. This product code construction gives an $[(r+1)^t,r^t,2^t]$-code with rate $R=\left(\frac{r}{r+1}\right)^t$. At the time, the authors stated that they believed this to be the largest rate attainable for a code with $t$ disjoint recovery sets, each with locality $r$.  This is very close to the bound rate from \eqref{eqn:TamoBargBounds} when $t=2$ but diverges from the bound for larger $t$. Second, in \cite{WangZhangLiu}, Wang et al. devise a parity check matrix construction giving rise to LRC($t$)s with rate $R=\frac{r}{r+t}$ for arbitrary $r$ and $t$.  These $\left [\binom{r+t}{t},\binom{r+t}{t}-\binom{r+t-1}{t-1}, t+1\right ]$-codes have better rate than product codes but even smaller minimum distance. The authors state that they believe their construction yields optimal rate for $t\leq r$.  Our literature search has not found any locally recoverable codes with $t\geq 2$ surpassing this rate. In what follows, we compare the rates and minimum distance of our most general example to these benchmarks.  In some cases we also compute the relative defect. 

\begin{remark}
In \cite{Bartoli2020}, Bartoli, Montanucci, and Quoos prove that codes with locality $(r_1, r_2,\dots, r_t)$ satisfy
  \begin{equation}\label{BMQBound}
  d\leq n-k-\left\lceil\frac{(k-1)t+1}{1+\sum_{i=1}^t r_i}\right\rceil + 2.
  \end{equation} 
  In all situations of this paper where this bound applies, we find that the bound in \cite{BhadaneThangaraj} is lower, so we compare to \eqref{BTBound} in what follows.

\end{remark}

\begin{remark}Many more complicated bounds have been proven for minimum distance, many incorporating field size.  See \cite{TBF16}, for example, and the survey \cite{Balaji2018} for a more comprehensive list of proven bounds. Further, some interesting constructions have been proven rate-optimal in specific cases, or to surpass the rate of the Wang et al. construction in \cite{WangZhangLiu} for certain $r$ and $t$ (for example \cite{KadheCalderbank} and binary simplex codes).
\end{remark}

This section attempts to shed some light on how different choices in code construction affect the parameters of the resulting codes, what parameters are attainable, and how these parameters compare to bounds and constructions in the literature.

\subsection{General Heuristics}
First, we consider the general code $C=C(V,B)$ with parameters described in Theorem \ref{thm:construction}.  Recall that $l$ is the degree of a divisor $D$ on $\Y(\FF_q)$, and $m$ is the dimension of the Riemann-Roch space $\mathcal{L}(D)$. For a fixed evaluation set $B$, it is clear that the dimension of $C$ increases (and the minimum distance of $C$ decreases) as $l$ increases up to its maximal value. The Riemann-Roch Theorem states that for a curve $\mathcal{Y}$ of genus $\gamma$, we have $m\geq l-\gamma+1$. The relationship of $m$ and $l$ depends on $D$ when $l< 2\gamma-1$, but when $l\geq 2\gamma-1$, we know that $m= l-\gamma+1$. Thus if all other parameters are fixed, the value of $m$ and therefore the dimension of $C$ will potentially be larger when $\gamma$ is smaller.  In our examples, we take $\Y=\PP^1$, so $\gamma=0$ and $m=l+1$.  Of course, $l$ and therefore $m$ are bounded by the number of points in $S$, and increasing the genus of $(\Y)$ can allow a larger number of points in $S$ by the Hasse-Weil bound.  Since $n=\left\vert B\right\vert=d_g\left\vert S \right\vert$, we may attain longer codes if $S$ is larger.  If all other parameters are fixed, this will decrease the rate but increase minimum distance.

Considering rate, we observe the following,
\begin{corollary}\label{cor:RatePC}
In the setting of Theorem~\ref{thm:construction}, if 
\begin{itemize}
\item the extensions $\FF_q(\Y_i)/\FF_q(\Y)$ are linearly disjoint, and
\item $d_{h_i}=d_{h_j}=r+1$ for all $i,j$,
\end{itemize}
then the rate $R$ satisfies
\[R\leq\left(\frac{r}{r+1}\right)^t.\]  
\end{corollary}
\begin{proof}
If $\gamma$ is the genus of the curve $\Y$, then the Riemann-Roch theorem implies that 
\[l+1\geq m\geq l-\gamma+1.
\]
Since the construction demands $l<\left\vert S\right\vert$ so that the evaluation map is injective, $m\leq \left\vert S \right\vert$.
\end{proof}
Therefore we see that when the fiber product construction is applied to yield codes with uniform locality $r$, it is not possible to create codes with rate surpassing that of the product code construction for the same availability and locality. The fiber product code construction is flexible, however, to allow for codes with larger minimum distance and to create varying locality across the recovery sets. 

In choosing curves $\Y_i$ and maps $h_i:\Y_i\to \Y$ for the fiber product, we know that the locality of $C$ will be determined by $d_{h_i}$.  All other things being equal, we should prefer small locality.  However, as we see in the formulas for parameters and the bounds above, this comes at a cost in rate and minimum distance.  Thus small locality must be balanced against efficiency and effectiveness in the code.  If the code is only to be used for local recovery, with global error correction never applied, large minimum distance is not useful, so we may wish to maximize rate given certain locality and availability conditions.  However, in some situations it may be that relatively large minimum distance is desirable to recover from a catastrophic event by global error correction or erasure repair.  Thus larger minimum distance may sometimes be desirable; in this case, larger relative minimum distance can be obtained by reducing the parameter $l$ to the minimum value of 0.

\subsection{LRC(2)s $C_{\HH_q}$ on $\HH_q$}

We return to the codes defined over $\F_{q^2}$ of Example \ref{ex:HermitianCurve} and Theorem \ref{thm:HermitianLRC2Dist}.  All parameters of these codes are dependent on the choice of $q$. We determined the parameters of $C_{\HH_q}$ for $q=p^h$ when $q\in\{2,3,5,7\}$ and $h\in\{1,2,3,4\}$. We compare to bounds on the minimum distance from \eqref{BTBound} as well as the relative defect from this bound.  This data is displayed in Table \ref{tab:HermitianTable}.

\begin{table}[ht]
\centering
\begin{footnotesize}
\begin{tabular}{lcccccc}
$q^2$ &($r_1,r_2$) &$n$ &$k$& $d$& upper bound on $d$ \eqref{BTBound}& $\frac{\textrm{bound}-d}{n}$\\
\hline
4 & (1, 2) & 6 & 2 & 4 & 4 & 0.0 \\
16 & (3, 4) & 60 & 12 & 38 & 46 & 0.1333 \\
64 & (7, 8) & 504 & 56 & 394 & 442 & 0.0952 \\
256 & (15, 16) & 4080 & 240 & 3602 & 3826 & 0.0549 \\
\hline
9 & (2, 3) & 24 & 6 & 14 & 17 & 0.1250 \\
81 & (8, 9) & 720 & 72 & 578 & 641 & 0.0875 \\
729 & (26, 27) & 19656 & 702 & 18254 & 18929 & 0.0343 \\
6561 & (80, 81) & 531360 & 6480 & 518402 & 524801 & 0.0120 \\
\hline
25 & (4, 5) & 120 & 20 & 82 & 97 & 0.1250 \\
625 & (24, 25) & 15600 & 600 & 14402 & 14977 & 0.0369 \\
15625 & (124, 125) & 1953000 & 15500 & 1922002 & 1937377 & 0.0079 \\
390625 & (624, 625) & 244140000 & 390000 & 243360002 & 243749377 & 0.0016 \\
\hline
49 & (6, 7) & 336 & 42 & 254 & 289 & 0.1042 \\
2401 & (48, 49) & 117600 & 2352 & 112898 & 115201 & 0.0196 \\
117649 & (342, 343) & 40353264 & 117306 & 40118654 & 40235617 & 0.0029 \\
5764801 & (2400, 2401) & 13841284800 & 5762400 & 13829760002 & 13835520001 & 0.0004 \\
\hline\\
\end{tabular}
\end{footnotesize}
\caption{Sample parameters for $C_{\mathcal{H}_q}$, an LRC($2$) over $\mathbb{F}_{q^2}$ with localities $r_1=q-1$ and $r_2=q$. The maximum possible minimum distance and relative defect from \eqref{BTBound} are also listed.}
\label{tab:HermitianTable}
\end{table}

We can also compute a formula for the defect and relative defect of $C_{\HH_q}$. Recall that $C_{\HH_q}$ is a $[q^3-q^2,q^2-2,q^3-2q^2+q+2]$-code with availability $2$ and locality $(q-1,q)$. Taking the bound \eqref{BTBound}, we find that $d\leq b$, where 
\[b=q^3-q-q^2+q+1-\left\lfloor\frac{q^2-q-1}{q-1} \right\rfloor-\left\lfloor\frac{q^2-q-1}{q^2-q} \right\rfloor=q^3-q^2-q+2.\]

Thus we can compute the defect to be $q^2-2q$ with a relative defect of $\frac{q^2-2q}{q^3-q}$, which approaches 0 as $q$ increases.  Thus these codes on the Hermitian curve have asymptotically good minimum distance.

\subsection{LRC(2) $C_{\X_q,l}$ on $\X_q=\HH_q\times_{\PP^1}\HH_q$}

For the codes $C_{\X_q,l}$ defined over $\F_{q^2}$ of Example \ref{ex:HermitianDistance2} and Theorem \ref{thm:THC Code}, we see that all parameters of these codes are also dependent on the choice of $q$.
%
%
In Table \ref{tab:my-Hqtableqs} we compare the parameters for $C_{l}$ over $\FF_{q^2}$ for different values of $q$. 

\begin{table}[ht]
\centering
\begin{tabular}{ccccccc}
$q$& $n$ & $l$ & $k$ &$d$ & bound from \eqref{BTBound} 
&$\frac{\textrm{bound}-d}{n}$\\
\hline
$4$ & $240$  & 0   & 12             & 142 & 226 
&0.35\\
& & 1   & 24  &         122 & 209       
&0.3625\\
& &2 & 36     &         102 & 192       
&0.375\\
& &3 & 48     &             82 & 175    
&  0.3875\\
& & 4 & 60    &           62   & 158       
&0.4\\
\hline
$5$ & $600$ & $0$ & 20   & 392  &     577                    
& 0.3083\\
& & 3   & 80  &        302                 & 499 
&0.3283\\
& &5 & 120     &         242              & 447 
&0.3416\\
\hline
7& 2352 & 0 & 42    &           1738            & 2305   
& 0.2410 \\
& &  3 & 168   &          1570         & 2155  
& 0.2487\\
& &  7 & 336   &          1346         & 1955   
&0.2589\\
\hline
$11$ & $14520$ & $0$ & 110   & 12014  &     14401           
& 0.1643 \\
& & 5   & 660  &      11354                 & 13791 
& 0.1678    \\
& &11 & 1320 &         10562              & 13059 
&0.1719  \\
\hline
13& 28392 & 0 & 156    &           24208        & 28225  
&0.1414 \\
& &  13 & 2184   &         21842         & 26015   
&0.1469 \\
\hline
\end{tabular}
\caption{Sample parameters for $C_{l}$, an LRC($2$) over $\mathbb{F}_{q^2}$ with locality $(r_1,r_2)=(q-1,q)$. 
}
\label{tab:my-Hqtableqs}
\end{table}

For $l=0$, the dimension of these codes is $k=q^2-q$ and the minimun distance $d=q^4-2q^3+3q+2$. Taking the bound from \eqref{BTBound}, we find that $d\leq q^4-2q^2+2$, and the relative defect is $\frac{2(q^2-q-1)}{q^3-q^2}$, which also approaches 0 as $q$ increases.  

\subsection{Parameters for LRC($t$) $C_{\mathcal{A}_{q,t},l}$ on Fiber Product of Artin--Schreier curves }

Here, we explore the parameter space of codes $C_{\mathcal{A}_{q,t},l}$ on the product of $t$ Artin--Schreier curves with points over $\F_{q^2}$ and $l$ the maximum degree in $y_0$ of functions leading to codewords. This family of codes is chosen for exploration because it can attain arbitrarily large availability $t$ (if extension degree of field of definition over prime field is allowed to increase) and arbitrary large locality ($r=p-1$ for any prime $p$).  This example family is not as general with regard to locality and availability as the product code and Wang et al. constructions, which allow for any $r$ and $t$ without increasing field size, but it is more general than many other concrete geometric constructions.  This example is not claimed to be optimal for the fiber product construction, only sufficiently adaptable to study parameters.

\subsubsection{Smallest Concrete Examples with $t=2$} The smallest non-trivial example in this case is a code of length 729 over the field $\FF_{81}$.  When $p$ is prime and $q=p^h$, the smallest $p$ allowing to non-constant functions in each $y_i$ with $1\leq i\leq t$ is $p=3$.  Since $t\leq h$, the smallest $h$ which allows multiple recovery sets is $h=2$.  Thus we may choose $l$ with $0\leq l \leq 74$ to be sure of positive minimum distance.  Theorem \ref{thm:ASDistance} gives the exact minimum distance for $0\leq l \leq 60$. In Table \ref{tab:AS p=3 t=2}, we give the parameters for the cases $l=0$, $l=60$, and $l=74$ and compare to the minimum distance bound in \cite{TamoBarg}.

\begin{table}[ht]
\centering
\begin{tabular}{ccccc}
$l$ & $k$ & rate & $d$ & bound on $d$ \eqref{eqn:TamoBargBounds} \\
\hline
0   & 4   & 0.006                            & 669 & 725              \\
60  & 244 & 0.334                            & 129 & 305              \\
74  & 300 & 0.412                            & 3*   & 207             
\end{tabular}
\caption{Sample parameters for $C_{\mathcal{A}_{q,t},l}$, an LRC($2$) over $\mathbb{F}_{81}$ with length $n=729$ and locality $r_1=r_2=2$.  The rate is bounded by $R\leq 0.533$. The listed distance when $l=74$ (marked with *) is  a lower bound from Theorem \ref{thm:ASCode} for the true minimum distance.}
\label{tab:AS p=3 t=2}
\end{table}

Larger codes can be created by increasing $p$, $h$, and/or $t$.  Letting $p=5$ and $h=t=2$ we obtain codes over $\FF_{625}$ with length $15625$. Here we may choose $l$ with $0\leq l \leq 593$ to be sure of positive minimum distance.  Theorem \ref{thm:ASDistance} gives the exact minimum distance for $0\leq l \leq 572$. In Table \ref{tab:AS p=5 t=2}, we give the parameters for the cases $l=0$, $l=572$, and $l=593$ and compare to the bounds in \eqref{eqn:TamoBargBounds}.

\begin{table}[ht]
\centering
\begin{tabular}{ccccc}
$l$ & $k$  & rate & $d$   & bound on $d$ \eqref{eqn:TamoBargBounds}\\
\hline
0   & 16   & 0.001         & 14845 & 15607            \\
572 & 9168 & 0.587         & 545   & 3595             \\
593 & 9504 & 0.608         & 20*    & 3154           
\end{tabular}
\caption{Sample parameters for $C_{\mathcal{A}_{q,t},l}$, an LRC($2$) over $\mathbb{F}_{625}$ with length $n=15625$ and locality $r_1=r_2=4$.  The rate is bounded by $R\leq 0.711$. The listed distance when $l=593$ (marked with *) is  a lower bound from Theorem~\ref{thm:ASCode} for the true minimum distance.}
\label{tab:AS p=5 t=2}
\end{table}

\subsubsection{Maximizing Rate for $t\geq 2$} As mentioned above, it may be of greatest interest to maximize the rate of LRC($t$)s for a given availability and locality.  
For $p=3, 5, 7$, we construct codes with $t=2, 3, 4$ over field $\F_{q^2}$ where $q=p^t$, the minimum field extension allowing $t$ recovery sets in this method.  
To obtain codes of large rate, we choose the maximum $l$ so that the minimum distance is guaranteed to be positive. 
The parameters of the codes arising from these choices are given in Table \ref{tab:AS p=3,5,7 t=2,3,4}.  In each case, we give a range where the minimum distance $d$ must lie based on the lower bound from Theorem \ref{thm:ASCode} and an upper bound from Theorem \ref{thm:ASDistance}.

\begin{table}[ht]
\centering
\begin{footnotesize}
\begin{tabular}{lccccccc}
$(p,t,l)$      & $q^2$   & $r$ & $n$         & $k$        & range for $d$      & rate  & rate bound \eqref{eqn:TamoBargBounds} \\
\hline
(3,2,74)       & 81      & 2   & 729         & 300        & {[}3,129{]}        & 0.415 & 0.533                                    \\
(3,3,700)      & 729     & 2   & 19683       & 5608       & {[}27, 1539{]}     & 0.285 & 0.457                                    \\
(3,4,6451)     & 6561    & 2   & 531441      & 103232     & {[}54, 17793{]}    & 0.194 & 0.406                                    \\
\hline
(5,2,593)      & 625     & 4   & 15625       & 9504       & {[}20, 545{]}      & 0.608 & 0.711                                    \\
(5,3,15398)    & 15625   & 4   & 1953125     & 985536     & {[}25, 19025{]}    & 0.505 & 0.656                                    \\
(5,4, 389122)  & 390625  & 4   & 244140625   & 99615488   & {[}375, 626625{]}  & 0.408 & 0.618                                    \\
\hline
(7,2, 2329)    & 2401    & 6   & 117649      & 83880      & {[}28, 1449{]}     & 0.712 & 0.791                                    \\
(7,3, 116911)  & 117649  & 6   & 40353607    & 25252992   & {[}294, 101479{]}  & 0.626 & 0.750                                    \\
(7,4, 5757938) & 5764801 & 6   & 13841287201 & 7462288944 & {[}343, 6593489{]} & 0.539 & 0.720                                   
\end{tabular}
\end{footnotesize}
\caption{Sample parameters for $C_{\mathcal{A}_{q,t},l}$, an LRC($t$) over $\mathbb{F}_{q^2}$, where $q=p^t$. Locality $r=p-1$ is the same for each recovery set.  We have chosen $l$ here to maximize dimension.}
\label{tab:AS p=3,5,7 t=2,3,4}
\end{table}

\subsubsection{Maximizing Minimum Distance for $t\geq 2$} If instead it is of interest to maximize the minimum distance of LRC($t$)s for a given availability and locality, this can be done by choosing $l=0$.  For $p=3, 5, 7$, we construct codes with $t=2, 3, 4$ over field $\F_{q^2}$ where $q=p^t$, the minimum field extension allowing $t$ recovery sets in this method.  The parameters of the codes arising from these choices are given in Table \ref{tab:ASDist l=0 p=3,5,7 t=2,3,4}.  In each case, we know the exact minimum distance $d$ from Theorem \ref{thm:ASDistance}, which we compare to the bound on minimum distance from \eqref{eqn:TamoBargBounds} and compute the relative defect.  

\begin{table}[ht]
\centering
\begin{footnotesize}
\begin{tabular}{lcccccccc}

$p$ & $t$ &$r$ &$n$ &$k$ &$d$ & bound on $d$ \eqref{eqn:TamoBargBounds} & $\frac{\textrm{bound} - d}{n}$\\
\hline
3 & 2 & 2 & 729 & 4 & 669 & 725 & 0.0768 \\
3 & 3 & 2 & 19683 & 8 & 18927 & 19672 & 0.0378 \\
3 & 4 & 2 & 531441 & 16 & 522585 & 531415 & 0.0166 \\
\hline
5 & 2 & 4 & 15625 & 16 & 14845 & 15607 & 0.0488 \\
5 & 3 & 4 & 1953125 & 64 & 1924775 & 1953044 & 0.0145 \\
5 & 4 & 4 & 244140625 & 256 & 243201625 & 244140289 & 0.0038 \\
\hline
7 & 2 & 6 & 117649 & 36 & 114149 & 117609 & 0.0294 \\
7 & 3 & 6 & 40353607 & 216 & 40100767 & 40353352 & 0.0063 \\
7 & 4 & 6 & 13841287201 & 1296 & 13824809481 & 13841285651 & 0.0012 \\
\hline\\
\end{tabular}
\end{footnotesize}
\caption{Sample parameters for $C_{\mathcal{A}_{q,t},l}$, an LRC($t$) over $\mathbb{F}_{q^2}$, where $q=p^t$. Locality $r=p-1$ is the same for each recovery set.  We have chosen $l=0$ here to maximize minimum distance.}
\label{tab:ASDist l=0 p=3,5,7 t=2,3,4}
\end{table}

\subsubsection{Rate for Increasing $t$ and Fixed Locality} Let $p$ (and thus locality $r=p-1$) be fixed and set $t=h$ to maximize the number of recovery sets for each field size.  For each $t$, let $l$ take on the maximum value guaranteeing positive minimum distance in Theorem \ref{thm:ASCode}. In Figure~\ref{fig:vary_field} we graph the rates of codes $C_{\mathcal{A}_{q,t},l}$ over $\F_{q^{2t}}$ for $t\in \{2,3,\dots 10\}$, increasing field size as well as number of recovery sets. We also graph the proven and conjectured rate bounds for codes with this availability and locality from \cite{TamoBarg}. The lengths of the codes are quite large, so we omit the accompanying tables. We find that the rate of the constructed codes is close to the product code rate for all $t$, growing extremely close as $t$ increases.  When we examine the minimum distance of these codes, we find that, in all cases except $(p,t)=(3,2)$, the minimum distance of the constructed code is larger than that of the corresponding product code and the Wang et al. construction.  However, this minimum distance comes at a cost of greater length and working over a larger field.

\begin{figure}[ht]
\centering
\includegraphics[width=3in]{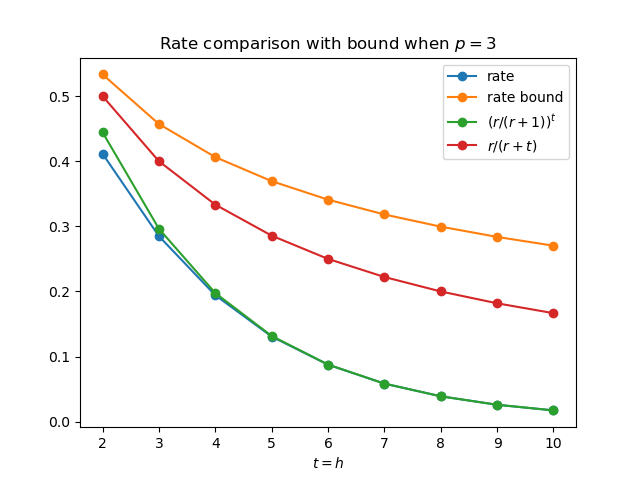}\includegraphics[width=3in]{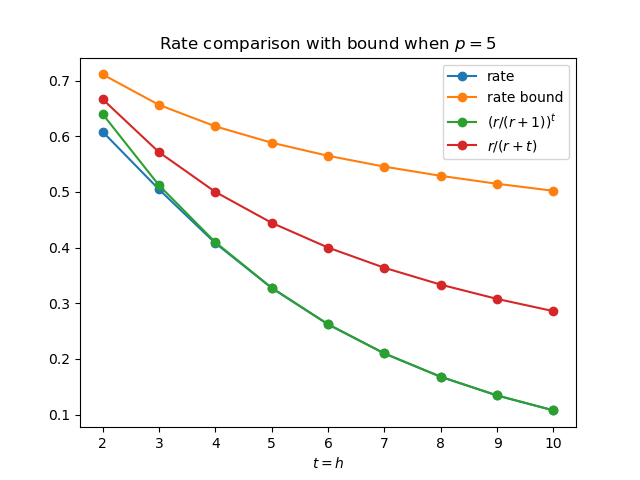}
\includegraphics[width=3in]{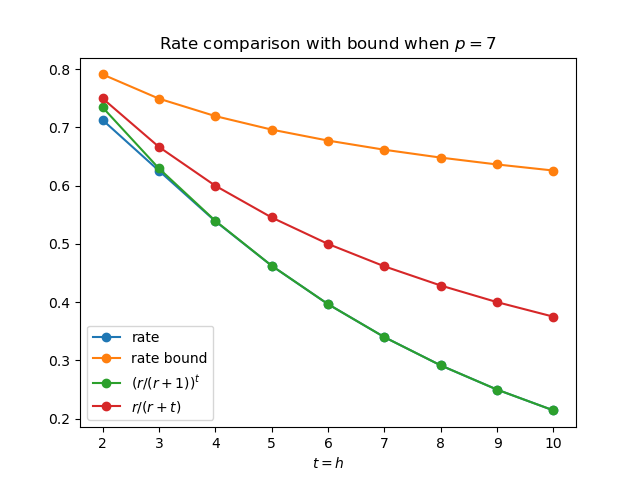}\includegraphics[width=3in]{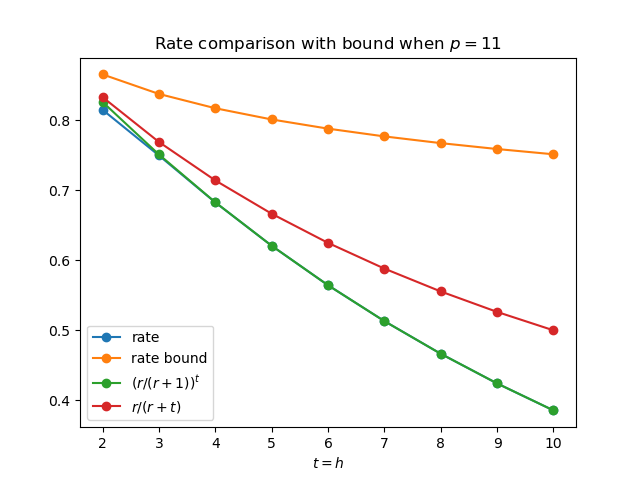}

\caption{Rate of codes $C_{\mathcal{A}_{p^{t},t},l}$, for $p=3$, $p=5$, $p=7$, and $p=11$, where $l$ is maximized for guaranteed positive minimum distance, defined over $\F_{p^{2t}}$ with $t$ recovery sets for $2\leq t\leq 10\}$. Also plotted are the bound on rate \eqref{eqn:TamoBargBounds}, and rates of product code and Wang et al. constructions (both defined over $\mathbb{F}_2$).  }
\label{fig:vary_field}
\end{figure}

One might wonder what happens when field size $q^2$ and locality $r$ are fixed, but the number $t$ of factor curves (and recovery sets) is increased.  In Figure~\ref{fig:p=3,5 t=2 to 10}, we graph the rate of the code $C_{\mathcal{A}_{p^{10},t},l}$ over $\F_{p^{20}}$, where $l$ is maximized for guaranteed positive minimum distance, as well as the proven rate bound from \eqref{eqn:TamoBargBounds} and the rates of the Tamo-Barg product construction and Wang et al. construction, for $t \in\{ 2, 3, \dots 10\}$. The lengths of the codes are quite large, so we omit the accompanying tables. We find that the rate of the fiber product codes is extremely close to the product code construction bound, matching up to at least 4 decimal places in each case.  The lower bound on minimum distances of these codes are also larger than that of the product code and Wang et al. construction when $(p,t)$ is not $(3,2)$.  This larger minimum distance comes at a cost of much greater length, however; the relative minimum distance of the fiber product codes is less than either of the other constructions.

\begin{figure}[ht]
\centering
\includegraphics[width=3in]{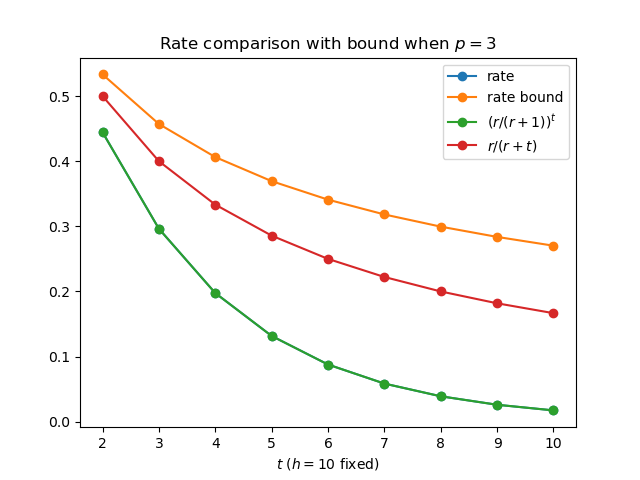}\includegraphics[width=3in]{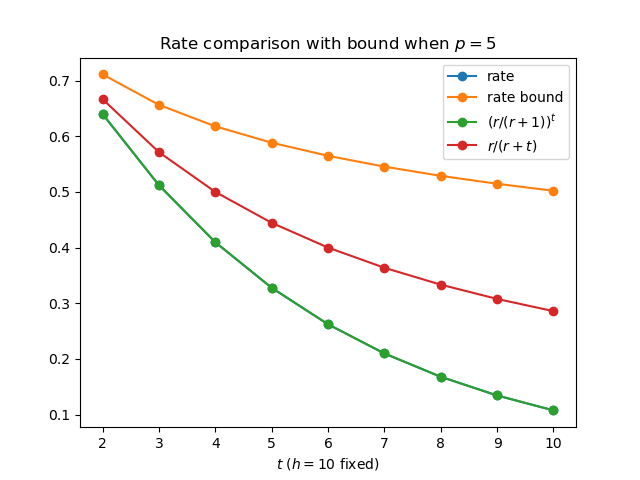}
\includegraphics[width=3in]{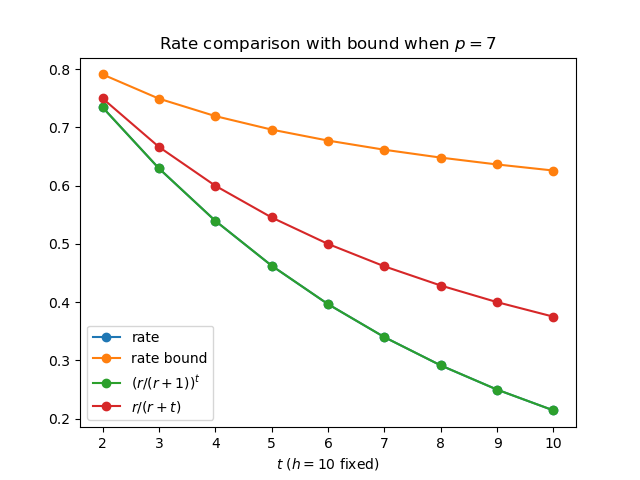}\includegraphics[width=3in]{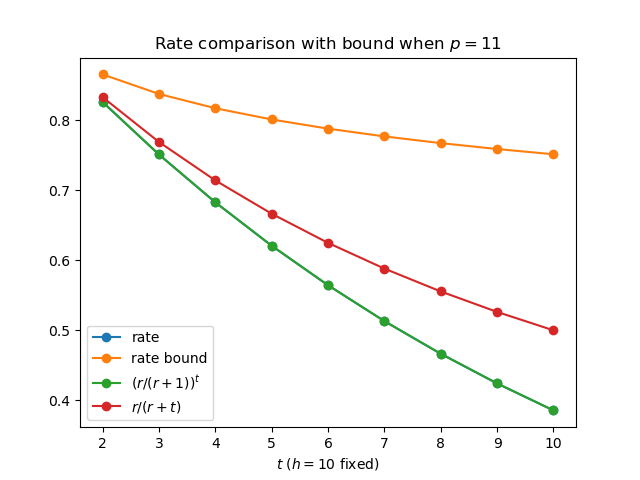}

\caption{Rate of codes $C_{\mathcal{A}_{p^{10},t},l}$, for $p=3$, $p=5$, $p=7$, and $p=11$, where $l$ is maximized for guaranteed positive minimum distance, defined over $\F_{p^{20}}$ with $t$ recovery sets for $2\leq t\leq 10\}$. Also plotted are the bound on rate \eqref{eqn:TamoBargBounds}, and rates of product code and Wang et al. constructions (both defined over $\mathbb{F}_2$). Code rate is visually indistinguishable from the rate of corresponding product code. }
\label{fig:p=3,5 t=2 to 10}
\end{figure}


\subsubsection{Rate for Fixed $t$ as Locality Increases} Finally, we consider the parameters of the codes $C_{\mathcal{A}_{p^{t},t},l}$ for various fixed values of $t$ as the prime $p$ increases and $l$ is chosen to maximize rate while guaranteeing positive minimum distance. Notice that the field size in each case is $\FF_{p^{2t}}$ so this increases with $p$ and $t$.  In Figure \ref{fig:increase_prime}

\begin{figure}[ht]\label{fig:increase_prime}
\centering
\includegraphics[width=3in]{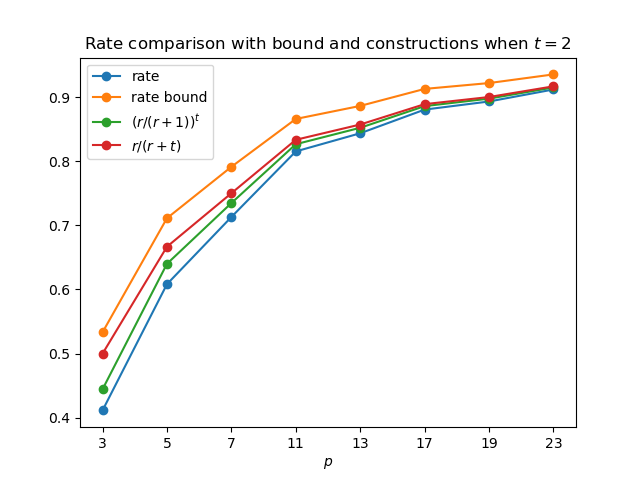}\includegraphics[width=3in]{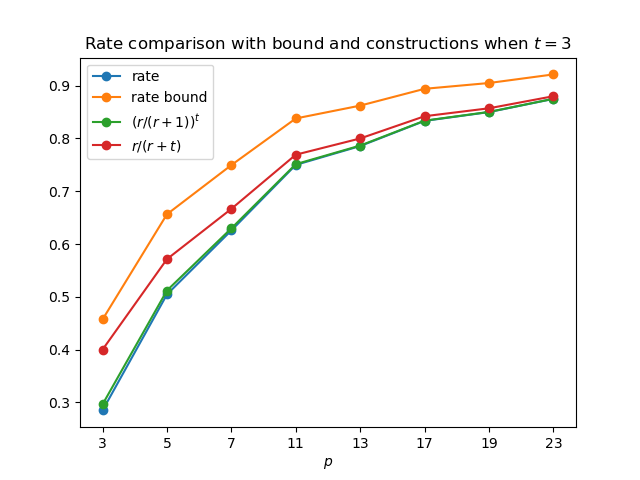}
\includegraphics[width=3in]{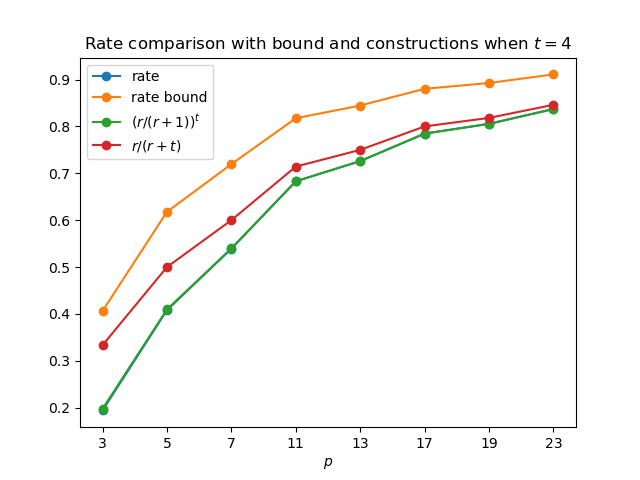}\includegraphics[width=3in]{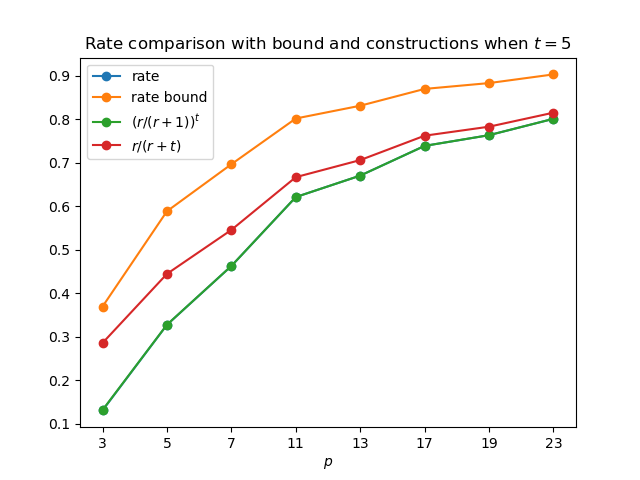}

\caption{Rate of codes $C_{\mathcal{A}_{p^{t},t},l}$, for $p\in\{3, 5, 7, 9, 11, 13, 17, 23\}$, where $l$ is maximized for guaranteed positive minimum distance, defined over $\F_{p^{2t}}$, with $t$ recovery sets for $0\leq t\leq 4$. Each recovery set has locality $p-1$. Also plotted are the bound on rate \eqref{eqn:TamoBargBounds}, and rates of product code and Wang et al. constructions (both defined over $\mathbb{F}_2$). Code rate is visually indistinguishable from the rate of corresponding product code when $t=4$ and $t=5$. }
\label{fig:p=3 to p=23}
\end{figure}

\subsubsection{Comparison with Product Code Rate as Locality Increases}

We now consider the actual formula for the rate of $C=C_{\mathcal{A}_{p^{t},t},l}$ when $l$ is chosen for maximal rate with positive minimum distance.  Recall that the length and dimension of $C$ are given by
\[k=q^2-\left\lfloor\frac{t(p^{t-1})(p-2)(p^t+1)+1}{p^t}\right\rfloor (p-1)^t,\]
\[n=p^tq^2.\]
With some simplification, this gives a rate at least
\[R\geq \frac{1}{p^{3t}}\left(p^{2t}(p-1)^t-\frac{t(p-2)(p^t+1)(p-1)^t}{p}-\left(\frac{p-1}{p}\right)^t+(p-1)^t\right).\]
As $p$ increases, we see that 
\[R\geq \left(\frac{p-1}{p}\right)^t-\mathcal{O}\left(\frac{1}{p}\right).\]
Thus at $p$ increases, the rate of $C$ approaches 1.  Further, the rate of $C$ is asymptotically the same as the rate of the corresponding product code (here, $r=p-1$, so the product code with matching locality and availability would have rate $\frac{(p-1)^t}{p^t}$).  Asymptotically, the Wang et al. construction grows at the same rate as well. Again, the product code and Wang et al. constructions both have the advantage of smaller field size.



\bibliography{bibliography}
\bibliographystyle{plain}
\end{document}